%% file: main.tex
\documentclass[]{IEEEtran}

\usepackage{wrapfig,enumerate}
\usepackage{cite}

  \usepackage{graphicx}
\usepackage{epstopdf}

\usepackage{amsmath,amsthm,amssymb}
\usepackage{color,url}

\usepackage{comment}

\def\E{\mathbb{E}}

\def\T{\mathcal{T}}

\usepackage[tight,footnotesize]{subfigure}

\newtheorem{theorem}{Theorem}

\newtheorem{lemma}{Lemma}

\begin{document}

\title{Fair Coexistence of Scheduled and Random Access Wireless Networks: Unlicensed LTE/WiFi}

\author{Cristina Cano$^1$, Douglas J. Leith$^2$, Andres Garcia-Saavedra$^3$, Pablo Serrano$^4$\\
\small
$^1$Inria Lille-Nord Europe, France,
$^2$Trinity College Dublin, Ireland,
$^3$NEC Labs Europe, Germany,
\\
$^4$Universidad Carlos III, Madrid.

}

\maketitle

\IEEEpeerreviewmaketitle
\input{abstract}
\input{intro}
\input{related}
\input{implications}
\input{faircoexistence}

\input{unsat}

\input{carriersense}

\input{scope}
\input{conclusions}

\bibliographystyle{IEEEtran}

\bibliography{references}

\end{document}

%% file: abstract.tex
\begin{abstract}

We study the fair coexistence of scheduled and random access transmitters sharing the same frequency channel.
Interest in coexistence is topical due to the need for emerging unlicensed LTE technologies to coexist fairly with WiFi. However, this interest is not confined to LTE/WiFi as coexistence is likely to become increasingly commonplace in IoT networks and beyond 5G.  In this article we show that mixing scheduled and random access incurs and inherent throughput/delay cost, the cost of heterogeneity.  We derive the joint proportional fair rate allocation, which casts useful light on current LTE/WiFi discussions.  We present experimental results on inter-technology detection and consider the impact of imperfect carrier sensing.

\end{abstract}


%% file: intro.tex
\section{Introduction}\label{sec:introduction}


In this paper we study the fair coexistence of scheduled and random access transmitters in the same frequency band.  Scheduled approaches transmit at regular instants of time (\emph{slot/frame/subframe} boundaries) whereas random access methods use carrier sensing to divide time up into variable-size slots.  We focus on the resulting MAC layer interactions and on joint MAC design for coexistence.  Our main contributions are the following: \emph{(i)} we show that mixing scheduled and random access incurs an inherent throughput/delay cost, which we refer to as the cost of heterogeneity, \emph{(ii)} we develop a joint throughput model for scheduled and random access transmitters  sharing the same band, \emph{(iii)} we derive the joint proportional fair rate allocation and \emph{(iv)} we present experimental measurements demonstrating the impact of imperfect carrier sensing by random access transmitters and show that our analytic results can be extended to encompass this.




While fair coexistence of scheduled and random access transmitters is of fundamental interest, it is particularly topical due to the 
current interest in operating LTE in unlicensed bands where WiFi is already widely deployed.    Regulators require mobile cellular operators to show that LTE, which is a scheduled protocol, can coexist in a \emph{fair} way with existing WiFi networks, which use random access~\cite{flore2014slides}.  In this context traditional power control solutions are of limited use and the requirement is to take into account the MAC layer interactions between the scheduled and random access approaches.  

Two main LTE mechanisms for coexistence with WiFi are presently under consideration.  Namely, \emph{Listen Before Talk} with \emph{Load Based Equipment} (LBT/LBE) and \emph{Carrier Sensing and Adaptive Transmission} (CSAT)~\cite{rahman2011license,qualcomm2014whitepapers}.   LBT/LBE uses carrier sensing and sends a reservation signal to grab the channel from WiFi.   In contrast, CSAT schedules transmissions according to a specified duty-cycle, oblivious to the channel status when a transmission is scheduled to start.  We will see that these two approaches are indeed two fundamental ways to ensure that a scheduled network has reasonable chances to transmit when sharing a channel with random access transmitters.  Further, our results establish that these two approaches can be operated in a proportional fair manner and show how this can be achieved, thereby providing significant input into current discussions on their ability to ensure fair coexistence with WiFi.


We note that interest in fair coexistence is not confined to LTE/WiFi, but also includes coexistence of WiFi and the TDMA access of Zigbee~\cite{zhang2011enabling} as well as WiFi and WiMaX~\cite{berlemann2006unlicensed,kim2011coexistence,bian2014addressing}.  It is also likely to be an important issue in the Internet of Things (IoT) context, where (i) Time-Slotted Channel Hopping (TSCH) protocols may be expected to coexist with random access approaches, both of which are defined in the IEEE 802.15.4e-2015 standard~\cite{IEEE802154e} and (ii) protocols such as the upcoming IEEE 802.11ah~\cite{IEEE80211ah} will need to coexist with Low Power Wide Area (LPWA) networks such as SigFox and LoRa~\cite{palattella2016internet}. More generally, we expect this kind of heterogeneity to become increasingly commonplace in the 5G era and beyond given the expected opportunistic use of spectrum and the growing range of network access technologies.


%% file: related.tex
\section{Related Work}\label{sec:related}



Coexistence among different technologies has traditionally been studied from an interference point of view, especially when coexisting devices have very different capabilities such as in the case of coexistence among WiFi and Bluetooth/Zigbee~\cite{friedman2013power,5672592,tytgat2012avoiding}.  However, taking into consideration the interactions among heterogeneous channel access mechanisms, in particular between scheduled (TDMA-like) and random-access mechanisms, allows for new insight and more scope for ensuring fair coexistence.  Previous work on coexistence of scheduled and random access mechanisms has considered WiFi and the TDMA access of Zigbee~\cite{zhang2011enabling} plus WiFi and WiMaX~\cite{berlemann2006unlicensed,kim2011coexistence,bian2014addressing}.  However, this work does not aim at providing formal fairness guarantees.  Recently, coexistence of WiFi and LTE has started to attract considerable interest.  The risks of employing legacy LTE in unlicensed bands without proper access control that ensures fair coexistence has been made quite evident in e.g.~\cite{cavalcante2013performance} (via simulations),~\cite{ltevswifi-01} (via analysis) or~\cite{lteu-exps} (via experiments).  The 3GPP's study on LTE/WiFi coexistence~\cite{3gppstudy} shows that the presence of unlicensed LTE networks may degrade the performance of existing 802.11 stations if coexistence protocols are not efficient. However, in this study the implementation details of the coexistence mechanisms used are not specified.  Nokia, Qualcomm and Huawei have presented their own white papers on the topic~\cite{nokia, qualcomm,huawei} showing satisfactory results. However, once again, details of the implemented access mechanisms (LBT and CSAT) and simulation models used in these papers are not public.    Fair coexistence of LBT has been studied in~\cite{ning2012unlicensed,hajmohammad2013unlicensed,liu2014small}. However, in these works the WiFi models used lack collisions and idle periods.  Fairness has also been studied in~\cite{ccano-icc} for a simplified version of the LBT scheme and without consideration of collisions between both technologies. Recently,~\cite{guan4cu} has studied how to jointly determine the channel selection, carrier aggregation and fractional spectrum access for CSAT so that the impact to WiFi throughput is {no more than that of another coexisting WiFi network}. However, they do not consider the inherent heterogeneity cost and the resulting model complexity does not allow for explicit solution.   The present paper substantially extends our initial findings in~\cite{ccano-icc-2016}, being both more general and taking account of important aspects such as non-saturated stations and imperfect inter-technology signal detection. 

%% file: implications.tex
\section{Implications of Heterogeneity}\label{sec:preliminaries}

Our interest is in coexistence of scheduled and random access networks in the same frequency band.  In this section we begin by considering the consequences for scheduled transmitters of being constrained to transmit at fixed slot times.

\subsection{Idle Channel Probability at Periodic Slot Boundaries}\label{sec:preliminaries-pidle_small}

Intuitively, when the random access transmitters are making efficient use of the channel, so leaving only a small amount of idle time, we expect that the probability of a scheduled transmitter finding the channel idle at the start of an admissible transmission slot will be small.   We formalise this intuituion as follows.

Consider a set of transmitters $A$ that are constrained to transmit in pre-defined time slots $[(j-1)\delta,j\delta)$, $j=1,2\dots$ each of duration $\delta$ (scheduled transmitters).   This might, for example, correspond to a network where the fixed time slots arise due to the use of a TDMA scheduler which divides time into slots and then schedules transmissions in these slots.  Suppose now that these transmitters $A$ share the radio channel with a set of transmitters $B$ (random access transmitters) which transmit during intervals $[T_k,T_k+\Delta_k)$, $k=1,2\dots$, with $T_{k+1} \ge T_k+\Delta_k$ and $\Delta_k$ the duration of the $k$'th transmission.   The start times $\{T_k\}$ are random variables that need not be synchronised with start times $(j-1)\delta$ of the pre-defined time slots used by transmitters $A$.   
The setup is illustrated schematically in Fig.~\ref{fig:slots}.   

We begin by asking for what fraction of $A$ slot start times $\{(j-1)\delta$,  $j=1,2\dots\}$ the channel is idle (i.e. there are no $B$ transmissions in progress).   This provides a measure of the transmission slots where the transmitters $A$ can schedule transmissions without interfering with the transmitters $B$.   Let $\T:=\cup_{k=1,2,\dots} [T_k,T_k+\Delta_k)$ denote the aggregate time occupied by $B$ transmissions and define random variable $X_j$ that takes value 1 when $(j-1)\delta \in \T$ and $0$ otherwise.   We are interested in the value of $p_{\rm idle}:=\lim_{J\rightarrow\infty} \frac{1}{J}\sum_{j=1}^J (1-X_j)$.

We can think of the $A$ transmitters as periodically sampling the channel at times $\{(j-1)\delta\}$ and $p_{\rm idle}$ as the probability that the channel is idle when they sample it.  Assuming that the start times and durations $\{(T_k,\Delta_k),k=1,2,\dots\}$ form a mixing process and that the sampling is not perturbing this process, then by~\cite[Theorem 2]{baccelli2006role} the NIMASTA property (a generalisation of PASTA) holds and $p_{\rm idle}$ is equal to the fraction of time the channel is idle \emph{i.e.} $p_{\rm idle}=\lim_{K\rightarrow\infty}\frac{\sum_{k=1}^{K-1}T_{k+1}-(T_k+\Delta_k)}{T_{K}}$.   


\begin{figure}
\centering
\includegraphics[width=0.6\columnwidth]{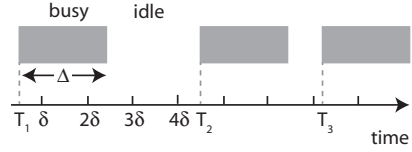}
\caption{Illustrating transmission slots with variable timing and with fixed timing of period $\delta$ (scheduled transmitter). The shaded rectangles indicate variable timing transmissions that start at times $T_k$, $k=1,2,\dots$ and are of duration $\Delta_k=\Delta$ (random access transmitter). }\label{fig:slots}
\end{figure}

\subsubsection{Example: CSMA/CA}\label{sec:ex}
Suppose that the scheduled transmitters $A$ are silent (we will relax this shortly) and the random access transmitters $B$ consist of $n$ stations using CSMA/CA.   Letting $\tau_i$ denote the probability that station $i$ transmits on a MAC slot then  $p_{\rm s}=\sum_{i=1}^n\tau_i\prod_{k=1,k\neq i}^n (1-\tau_k)$ is the probability of a successful transmission and $p_{\rm c}=1-p_{\rm s}-p_{\rm e}$, with $p_{\rm e}=\prod_{i=1}^n (1-\tau_i)$, is the probability of a collision between transmissions.    Let $T_b$ denote the duration of a successful transmission, including the MAC ACK, and $T_{\rm fra}$ the duration of a data frame without corresponding ACK i.e. of a colliding transmission.   Hence, $\Delta_k=T_{\rm b}$ for successful transmissions and $\Delta_k=T_{\rm fra}$ for collisions.   When the transmissions form a renewal process we then have that,
\begin{equation}\label{eq:pidle}
 p_{\rm idle} = 1-\frac{p_{\rm s} T_{\rm b}+ p_{\rm c} T_{\rm fra} }{\E[M]},
\end{equation}
where $\E[M]$ is the average MAC slot duration of CSMA. 

To proceed we insert typical 802.11ac~\cite{IEEE80211ac} values into (\ref{eq:pidle}).    Namely,

\small
\begin{align}
T_{\rm fra}&=T_{\rm plcp}+ \left\lceil \frac{L_{\rm s} + n_{\rm agg}(L_{\rm del} + L_{\rm mac-h} + D) + L_{\rm t}}{n_{\rm sym}} \right\rceil T_{\rm s},\nonumber\\
T_{\rm ack}&=T_{\rm plcp}+\left\lceil \frac{L_{\rm s} + L_{\rm ack} + L_{\rm t}}{n_{\rm sym}} \right\rceil T_{\rm s},
T_{\rm b}=T_{\rm fra} + \rm{SIFS} + T_{\rm ack},\nonumber
\end{align}
\normalsize
where $n_{\rm sym}$ is the number of bits per OFDM symbol, $T_{\rm s}$ is the symbol duration, $n_{\rm agg}$ is the number of packets aggregated in a transmission and the values of the various parameters are specified in Table~\ref{tbl:param}. We also set $\E[M] = \sigma p_{\rm e} + (p_{\rm s} +p_{\rm c} )(\rm{DIFS}+T_{\rm b})$, where $\sigma$ is the duration of a PHY slot~\cite{checco2011proportional}.

Fig.~\ref{fig:p_idle_onlysampling}(a) shows $p_{\rm idle}$ calculated using (\ref{eq:pidle}) vs the number of packets aggregated in a transmission (effectively varying $T_{\rm fra}$) for a WLAN with $n=1$ and $n=3$ stations and $\delta$ deterministic (referred as ``Periodic'') and equal to $100$~ms. The parameters used are detailed in Table~\ref{tbl:param}, $\tau_i, i=1,...,n$ set to 1/16 and MCS configured to $64$-QAM 5/6 with $20$ MHz channel width.  Also shown in Fig.~\ref{fig:p_idle_onlysampling}(a) is the measured fraction of periodic slots $\{(j-1)\delta, j=1,2,\dots\}$ obtained by numerical simulation and, as expected, it can be seen that they are in good agreement.   

It can be seen from Fig.~\ref{fig:p_idle_onlysampling}(a) that the value of $p_{\rm idle}$ is relatively small (in general, $<50\%$ and below $5\%$ for larger WLAN packet sizes), indicating that relatively few non-colliding transmission slots are available for use by the scheduled transmitters $A$.  


\begin{table}[tb]
\centering
\caption{Parameters IEEE 802.11ac~\cite{IEEE80211ac}}\label{tbl:param}
\begin{tabular}{c|c} 
Slot Duration ($\sigma$) & $9$~$\mu$s \\ \hline
DIFS & $34$~$\mu$s\\ \hline
SIFS & $16$~$\mu$s\\ \hline
PLCP Preamble+Headers Duration ($T_{\rm plcp}$) & 40~$\mu$s  \\ \hline
PLCP Service Field ($L_{\rm s}$) & 16 bits \\ \hline
MPDU Delimiter Field ($L_{\rm del}$) & 32 bits  \\ \hline
MAC Header ($L_{\rm mac-h}$) & 288 bits  \\ \hline
Tail Bits ($L_{\rm t}$) & 6 bits  \\ \hline
ACK Length ($L_{\rm ack}$) & 256 bits  \\ \hline
Payload ($D$) & $12000$ bits \\ \hline
\end{tabular}
\end{table}

\subsection{Cost of Heterogeneity}\label{sec:preliminaries-cost}
An important consequence of the fact that $p_{\rm idle}$ is typically small is that for transmitters $A$, which are restricted to transmit at periodic times $\{(j-1)\delta, j=1,2,\dots\}$, at the great majority of the potential transmission times competing transmissions $B$ are already in progress.   This means that for the transmitters $A$ to have a reasonable chance to transmit at the start of a slot boundary, they must either act: \emph{(i)} \emph{Preemptively:} transmitting at the start of a slot boundary regardless of the channel status, thus potentially causing collisions with transmitters $B$ or \emph{(ii)} \emph{Opportunistically:} grabbing the channel when empty and transmitting a reservation signal until the next slot boundary (assuming the transmitters $B$ can effectively detect A's transmissions\footnote{We will revisit this assumption later in Section~\ref{sec:carrier_sense}.}, then a reservation signal will make transmitters $B$ refrain from accessing the channel).  Note that the \emph{Preemptive} approach can be identified with the LTE CSAT approach and the \emph{Opportunistic} approach with LBT/LBE. Both cases incur a reduction 
of effective airtime since in \emph{(i)} additional collisions are generated, and so network throughput is lowered, while in \emph{(ii)} the reservation signal reduces the airtime available for data transmissions which again lowers network throughput\footnote{Note that in contrast to the airtime loss due to collisions, the reservation signal could be used to transmit control or other information.}. 
   That is, the heterogeneity of the transmission slots used by transmitters $A$ and $B$ necessarily incurs an overhead.   

We quantify this overhead in more detail later since it is technology-dependent, but for now we note that {provided transmitters $B$ can effectively detect transmitters $A$ (e.g. via carrier sensing), the throughput overhead can be reduced by increasing the duration of the transmissions by scheduled stations $A$.} This can be seen by noting that the overhead is then a per-transmission one (either a single collision or a single reservation signal is incurred per transmission).  Hence, increasing the duration of $A$ transmissions amortises this overhead over a larger amount of data and increases network throughput efficiency.    However, increasing the duration of $A$ transmissions will tend to increase the delay experienced by the $B$ transmissions since these now need to wait longer for $A$ transmissions to finish before they can start to transmit.   The overhead incurred by use of heterogeneous transmission slots can therefore also be expressed as a trade-off between network throughput and delay.   

\emph{In summary, heterogeneity necessarily incurs a per-transmission overhead which can be alleviated by increasing the duration of the scheduled transmissions provided that random access transmitters effectively detect those. In turn, that solution tends to increase the delay of the random access transmissions.}

\begin{figure}
\centering
\subfigure[Without $A$ transmissions.]{\includegraphics[width=0.7\columnwidth]{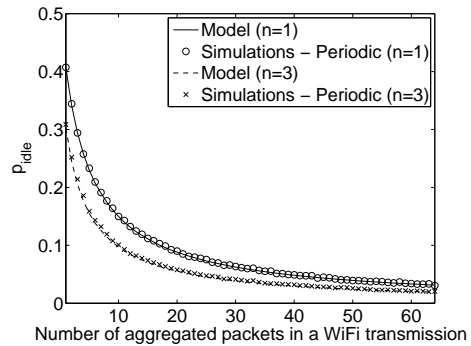}}
\subfigure[With $A$ transmissions (n=1).]{\includegraphics[width=0.7\columnwidth]{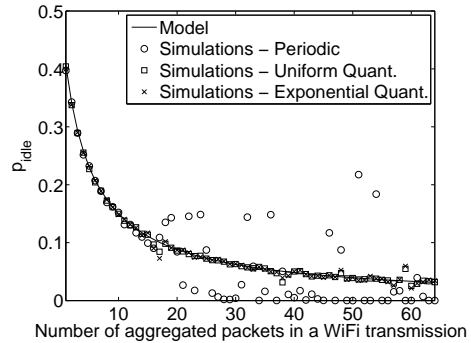}}
\caption{Probability of observing an idle channel at slot boundaries of fixed period $\delta=100$~ms in a channel occupied by an 802.11 WLAN with $n=1$ and $n=3$ WiFi stations.  Simulation results are averages of $100$ simulation runs with $50$ s time horizon.}
\label{fig:p_idle_onlysampling}
\end{figure}

\subsection{Example Revisited: Random Scheduled Starting Times}\label{sec:preliminaries-randomise}
Recall that in the previous example we assumed that the scheduled transmitters $A$ are silent.   We now relax this assumption.  The trickiest case is when scheduled transmissions $A$ can be detected by the random access transmitters $B$, e.g. via carrier sensing, and so the random access transmissions are coupled to the scheduled transmissions.  In this case the NIMASTA property does not hold.   Nevertheless, we show below that provided the time between scheduled transmissions is suitably randomised then simulations indicate that  the insight from the previous analysis generally remains valid.


We begin by highlighting the impact of coupling between the random access and scheduled transmissions via carrier sensing.
Fig.~\ref{fig:p_idle_onlysampling}(b) shows measurements of $p_{\rm idle}$ obtained by numerical simulation for a setup similar to that in Section~\ref{sec:ex} and can be directly compared with Fig.~\ref{fig:p_idle_onlysampling}(a).  The difference is that now there is an $A$ transmitter which starts transmitting at a slot boundary regardless of the channel status and keeps transmitting for a fixed duration $T_{\rm on}=50$~ms which is a multiple of the slot duration $\delta=1$~ms.  It then remains silent for a period $T_{\rm off}$ seconds.   

When $T_{\rm off}$ is deterministic and fixed at $T_{\rm off}=50$~ms (labelled ``Periodic'' in Fig.~\ref{fig:p_idle_onlysampling}(b)), it can be seen that $p_{\rm idle}$ exhibits quite complex behaviour as the duration of the WLAN transmissions is varied.   Further inspection confirms that this is associated with interactions between the $A$ and $B$ transmissions induced by detection of $A$ transmissions by the $B$ nodes.  Namely, due to carrier sensing $B$ transmissions are deferred during each $T_{\rm on}$ interval and then restart during the $T_{\rm off}$ interval. The $B$ transmission behaviour following restart is constrained (there can be no ongoing $B$ transmissions at the start of a $T_{\rm off}$ interval) and this leads to quantisation effects related to the number of complete $B$ transmissions that can be fitted into the $T_{\rm off}$ interval.  Observe that this quantisation effect is non-negligible even when $T_{\rm off}$ is relatively long ($T_{\rm off}$ is set to $50$~ms in Fig.~\ref{fig:p_idle_onlysampling}(b)).

For comparison,  Fig.~\ref{fig:p_idle_onlysampling}(b) shows the corresponding data when $T_{\rm off}$ is drawn randomly after each $T_{\rm on}$ interval according to uniform and exponential distributions with mean $50$~ms, minimum $10$~ms and rounded to a multiple of $\delta$ (labelled as ``Uniform Quant.'' and ''Exponential Quant.'', respectively).   It can be seen that {randomising $T_{\rm off}$ largely removes the quantisation effects and the measured $p_{\rm idle}$ is once again in good agreement with (\ref{eq:pidle})}. The analysis here indicates that it is probably preferable to randomise the duration of the  $T_{\rm off}$ intervals to avoid quantisation effects.

%% file: faircoexistence.tex
\section{Fair Coexistence}\label{sec:coexistence}


In this section we consider fairness for both of the fundamental coexistence approaches noted in the previous section (\emph{Preemptive} and \emph{Opportunistic}), taking a proportional fair approach. Since achieving fair coexistence is only of concern when we want to make intensive use of the network resources, we consider in the following carrier sense random access, as mechanisms without carrier sense are well known to perform poorly in these conditions, e.g.~\cite{yang2003delay}.  

\subsection{Throughput Model}\label{sec:appendix_model}



We begin by developing a throughput model when transmitters $A$ and $B$ share the same wireless channel. We consider the \emph{Preemptive} and \emph{Opportunistic} approaches in a unified fashion, to simplify both the model and the later fairness analysis.

\begin{figure}
\centering
\includegraphics[width=0.8\columnwidth]{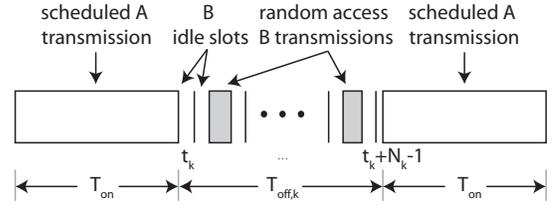}
\caption{Schematic showing scheduled/random access transmission timing.}\label{fig:slots2}
\end{figure}

We consider a single\footnote{The case of multiple scheduled transmitters will be discussed later in Section \ref{sec:discussion}.} scheduled transmitter ($A$) sharing the network with a set of $n$ transmitters $B$. Let $S_k$, $k=1,2,\dots$ denote the times when scheduled transmissions start.  Each scheduled transmission is of duration $T_{\rm on}$, so the silent/off interval between scheduled transmissions $k$ and $k+1$ is of duration $T_{{\rm off},k}:=S_{k+1}-S_k-T_{\rm on}$, see Fig.~\ref{fig:slots2}.  From the insights obtained in Section~\ref{sec:preliminaries}, we assume that random variables $T_{{\rm off},k}$, $k=1,2,\dots$ are i.i.d with mean $\bar{T}_{{\rm off}}:=\E[T_{\rm off}]$.  As we pointed out before, we also assume that transmitters $B$ use CSMA/CA. We further consider that transmitters $B$ are able to detect the channel as being busy during a scheduled transmission with no error; so no CSMA/CA station starts transmission during a $T_{\rm on}$ period, we will revisit this 
assumption later in Section~\ref{sec:carrier_sense}. Note that there may be a collision at the start of a $T_{\rm on}$ period when transmitter $A$ starts transmitting while a CSMA/CA transmission is already in progress.
\newline
\subsubsection{CSMA/CA MAC Slots}

During the $T_{{\rm off},k}$ period when transmitter $A$ is silent following the end of a $T_{\rm on}$ period, the random access stations perform their usual CSMA/CA mechanism.   This process partitions time into MAC slots which may be either an idle slot, of duration $\sigma$, or a busy slot, of duration $\Delta$ (for simplicity, we assume here that both successful transmissions and collisions between the $n$ transmitters are of the same duration), where $\Delta$ denotes the time to transmit a frame.   We index these MAC slots during the $T_{{\rm off},k}$ period by $t_k,t_{k}+1,\dots,t_{k}+N_k-1$, where the number of MAC slots $N_k:=t_{k+1}-t_k$ in the $T_{{\rm off},k}$ period is a random variable, see Fig.~\ref{fig:slots2}. Note that at the end of the $T_{{\rm off},k}$ period there will generally be a partial MAC slot, since the end of the $T_{{\rm off},k}$ period need \emph{not} be aligned with the CSMA/CA MAC slot boundaries, but $t_{k}+N_k$ indexes the last full MAC slot.  

\vspace{0.15cm}

\subsubsection{CSMA/CA Events}

Let $Z_{t,j}$ be a random variable which takes the value $1$ when a CSMA/CA station $j$ transmits in MAC slot $t$.  We assume that the $Z_{t,j}$, $t=1,2,\dots$ are i.i.d, $Z_{t,j}\sim Z_{j}$ and let $\tau_j:=\Pr(Z_j=1)$.  We also assume that the $Z_{t,j}$, $j=1,\dots,n$ are independent.   

Let $X_{t}$ be a random variable which takes the value $1$ when MAC slot $t$ is busy ($Z_{t,j}=1$ for at least one $j\in\{1,\dots,n\}$), and $0$ otherwise.  The $X_{t}$, $t=1,2,\dots$ are i.i.d, $X_{t}\sim X$, with $p_{\rm e}:=\Pr(X=0)=\prod_{i=1}^n(1-\tau_i)$.   
Since the $X_{t}$, $t=1,2,\dots$ are i.i.d and the $T_{{\rm off},k}$, $k=1,2,\dots$ are also i.i.d. the number $N_k$ of MAC slots in the $T_{{\rm off},k}$ periods $k=1,2,\dots$ are i.i.d, $N_k\sim N$.    The duration of MAC slot $t$ is $M_t:=\sigma+X_{t}(\Delta-\sigma)$.  The $M_t$, $t=1,2,\dots$ are i.i.d, $M_t\sim M$ with $\E[M]=\sigma p_{\rm e}+\Delta(1-p_{\rm e})$.

Let $Y_{t,j}$ be a random variable which takes the value $1$ when there is a successful (non-colliding) transmission by a CSMA/CA station $j$ in MAC slot $t$, and $0$ otherwise.    The $Y_{t,j}$, $t=1,2,\cdots$ are i.i.d, $Y_{t,j} \sim Y_j$, with $p_{succ,j}:=\Pr(Y_{j}=0)=\frac{\tau_j}{1-\tau_j}p_{\rm e}$.   
The number of successful transmissions in the $T_{{\rm off},k}$ period is $W_{k,j}:=\sum_{t=t_k}^{t_{k}+N_k-1} Y_{t,j}$ and the mean rate in bit/s of a CSMA/CA station $j$ is $s_{{\rm csma},j}:=\lim_{K\rightarrow\infty}\frac{\sum_{k=1}^K W_{k,j}}{\sum_{k=1}^K T_{\rm on}+T_{{\rm off},k}} D_j$, where $D_j$ is the number of data bits communicated by station $j$ in a successful transmission.

\vspace{0.15cm}

\subsubsection{CSMA/CA Throughput}

The $W_{k,j}$, $k=1,2,\cdots$ are i.i.d, $W_{k,j}\sim W_j$, and the $T_{{\rm off},k}$, $k=1,2,\cdots$ are also i.i.d, $T_{{\rm off},k}\sim T_{{\rm off}}$  (but note that $W_{k,j}$ and $T_{{\rm off},k}$ are not independent since the number of successful transmissions depends on the duration of the $k$'th \emph{off} period).   The $W_{k,j}$, $T_{{\rm off},k}$, $k=1,2,\cdots$ define a renewal-reward process and it follows that  $s_j =\frac{\E[W_{j}]}{T_{\rm on}+\E[T_{{\rm off}}]} D_j$.
We have that $\E[W_{j}] = \E[Y_{j}]\E[N]$ since $Y_{j}$ and $N$ are independent, and $\E[Y_{j}]=p_{succ,j}$.  It remains to determine $\E[N]$ (the average number of full CSMA/CA MAC slots in an \emph{off} period of transmitter $A$).

Let $\hat{T}_{{\rm off},k}=\sum_{t=t_k}^{t_{k}+N_k-1}M_t$.  That is, $\hat{T}_{{\rm off},k} \le T_{{\rm off},k}$ is the duration of that part of the $T_{{\rm off},k}$ period occupied by full CSMA/CA MAC slots \emph{i.e.} excluding any partial MAC slot at the end of the period that may take place when transmitter $A$ uses the \emph{Preemptive} approach.  It follows that $\E[N] = \frac{\E[\hat{T}_{{\rm off}}]}{\E[M]}$ since the $M$ is independent of $N$.   Hence,
\begin{align}
s_{{\rm csma},j} = \frac{p_{succ,j}}{\sigma p_{\rm e}+\Delta(1-p_{\rm e})}\frac{\E[\hat{T}_{{\rm off}}]}{T_{\rm on}+\bar{T}_{{\rm off}}} D_j.
\end{align}
Observe that $s_j:=\frac{p_{succ,j}}{\sigma p_{\rm e}+\Delta(1-p_{\rm e})}D_j$ is just the usual expression for the throughput of a CSMA/CA station (as in~\cite{checco2011proportional} for the case of 802.11), but that this is now scaled by $\frac{\E[\hat{T}_{{\rm off}}]}{T_{\rm on}+\bar{T}_{{\rm off}}}$.

\vspace{0.15cm}

\subsubsection{$\E[\hat{T}_{{\rm off}}]$}

To complete the expression for CSMA/CA throughput we require $\E[\hat{T}_{{\rm off}}]$. We show in the following how to compute this for the two mechanisms considered. \newline

\underline{Preemptive Approach}

Since in the \emph{Preemptive} approach, transmitter $A$ transmits regardless of the channel status, a scheduled transmission may start part way through a CSMA/CA MAC slot.   In this case we might approximate $\E[\hat{T}_{{\rm off}}]$ by $\E[T_{{\rm off}}]$, and we can expect this approximation to be accurate when $\E[T_{{\rm off}}]$ is sufficiently large that any partial CSMA/CA MAC slots can be neglected.  However, when $\E[T_{{\rm off}}]$ is smaller it is necessary to use a more accurate approximation for $\E[\hat{T}_{{\rm off}}]$.   We adopt the following.    When the start times $S_k$, $k=1,2,\dots$ of the scheduled transmissions satisfy the lack of anticipation property, \emph{e.g.} when the spacing $S_{k+1}-S_k$ is drawn from an exponential distribution~\cite{baccelli2006role}, then the transmissions from transmitter $A$ satisfy the PASTA property. This is in turn in line with the insights obtained in Section~\ref{sec:preliminaries-randomise}.   
Then, the probability that the start of a transmission from transmitter $A$ coincides with a CSMA/CA transmission is $p_{{\rm txA}} = \frac{(1-p_{\rm e})\Delta }{\E[M]}$.   Assuming that on average the start of an \emph{on} period that collides with a CSMA/CA transmission occurs half-way through the CSMA/CA transmission, then

\small
\begin{align}
\E[\hat{T}_{{\rm off}}]&=\E[T_{{\rm off}}](1-p_{{\rm txA}}) + (\E[T_{{\rm off}}] - \Delta/2)p_{{\rm txA}}\nonumber\\
&=\bar{T}_{{\rm off}} -c_1,
\end{align}
\normalsize
$c_1$ is the average airtime lost due to a partial CSMA/CA MAC slot before the start of an \emph{on} period (when the scheduled transmitter starts transmitting) with $c_1= (\Delta/2)p_{{\rm txA}}$.\newline

\underline{Opportunistic Approach}

In the \emph{Opportunistic} approach, we assume that the start of an \emph{on} period is aligned with a CSMA/CA MAC slot boundary since the scheduled network must in this case detect CSMA/CA transmissions and can then ensure this.   Therefore there are no partial MAC slots and 
\begin{align}
\E[\hat{T}_{{\rm off}}]&=\bar{T}_{{\rm off}}.
\end{align}
Equivalently, $\E[\hat{T}_{{\rm off}}]=\bar{T}_{{\rm off}}-c_1$ where the average airtime lost due to a partial CSMA/CA MAC slot before the start of an \emph{on} period is now $c_1:= 0$.  Also, the  probability that the start of an \emph{on} period coincides with a CSMA/CA transmission is just $p_{{\rm txA}} = 1-p_{\rm e}$, that is, the probability of having at least one CSMA/CA station transmitting in a given MAC slot.  
\vspace{0.25cm}

\subsubsection{Scheduled Network Throughput}

Let $r$ denote the transmit rate in bits/s of transmitter $A$.  When the start time of an \emph{on} period does not coincide with a CSMA/CA transmission then the error-free transmission of transmitter $A$ is of duration $T_{\rm on}$ \emph{i.e.} $rT_{\rm on}$ bits are transmitted.   When the \emph{on} start time coincides with a CSMA/CA transmission then we assume that the first fully or partially overlapping slots of the transmission are lost.  The precise behaviour differs for the \emph{Preemptive} and \emph{Opportunistic} approaches, as follows.\newline

\underline{Preemptive Approach}

On average the start of a transmission from transmitter $A$ that collides with a CSMA/CA transmission occurs half-way through the CSMA/CA transmission, and so on average the first $\Delta/2$ seconds of the transmission from the scheduled network are lost.  Assuming that partial overlap of a scheduled slot with a CSMA/CA transmission leads to loss of the whole slot, then $r(T_{\rm on}-\lceil\frac{\Delta}{2\delta}\rceil \delta)$ bits are transmitted by transmitter $A$, where $\delta$ denotes the duration of a pre-defined slot in the scheduled network.  It follows that the resulting throughput when using the \emph{Preemptive} approach is:
\begin{align}
s_{{\rm txA}} &=r\frac{T_{\rm on}(1-p_{{\rm txA}}) + (T_{\rm on}-\lceil\frac{\Delta}{2\delta}\rceil \delta)p_{{\rm txA}}} {T_{\rm on}+\bar{T}_{{\rm off}}}\nonumber\\
&=r\frac{T_{\rm on}-c_2} {T_{\rm on}+\bar{T}_{{\rm off}}},
\end{align}
where $c_2$ is the mean scheduled airtime during which collisions with the CSMA/CA transmitters occur, with $c_2=\lceil\frac{\Delta}{2\delta}\rceil \delta p_{{\rm txA}}$.\newline 

\underline{Opportunistic Approach}

Since in this case the scheduled transmissions are aligned with CSMA/CA MAC slots the duration of a collision between both networks is simply $\Delta$.  Additionally, since the transmitter $A$ has to transmit a reservation signal until the next subframe boundary of average duration $T_{\rm res}=\delta/2$, useful data transmission only occurs during $T_{\rm on}-T_{\rm res}$ and the number of bits transmitted at each \emph{on} period when it suffers from a collision with the CSMA/CA network is: $r(T_{\rm on}-\max(T_{\rm res}, \lceil\frac{\Delta}{\delta}\rceil \delta))$.  It follows that the throughput achieved by the scheduled network when using the \emph{Opportunistic} approach is:

\small
\begin{align}
s_{{\rm txA}} &=r\frac{T_{\rm on}-c_2}{T_{\rm on}+\bar{T}_{{\rm off}}},
\end{align}
\normalsize
where now $c_2=\max(T_{\rm res}, \lceil\frac{\Delta}{\delta}\rceil \delta)p_{{\rm txA}} + T_{\rm res}(1-p_{{\rm txA}})$ is the mean scheduled airtime during which collisions between scheduled and CSMA/CA transmissions occur and/or the scheduled transmitter sends a reservation signal.

\subsection{Proportional Fair Allocation}\label{sec:appendix_prop_fair_allocation}

We now derive the proportional fair rate allocation when transmitter $A$ (scheduled) and the $n$ nodes of type $B$ (CSMA/CA) share a channel and transmitter $A$ uses either the: (i) \emph{Preemptive} or (ii) \emph{Opportunistic} approach.

\begin{theorem}[Proportional Fair Rate Allocation]\label{tm:one}
The proportional fair rate allocation assigns the following fraction of channel airtime to full CSMA/CA MAC slots:
\begin{align}
\frac{\bar{T}_{{\rm off}}^* -c_1}{T_{\rm on}+\bar{T}_{{\rm off}}^*} 
&= \frac{n}{n+1}\label{eq:p1}
\end{align}
and the following fraction of airtime to the scheduled network:
\begin{align}
\frac{T_{\rm on}+c_1}{T_{\rm on}+\bar{T}_{{\rm off}}^*} 
&=1-\frac{\bar{T}_{{\rm off}}^* -c_1}{T_{\rm on}+\bar{T}_{{\rm off}}^*} 
=\frac{1}{n+1},\label{eq:p2}
\end{align}
where $\bar{T}_{{\rm off}}^*$ is the proportional fair mean off time between scheduled transmissions, $c_1$ is the average airtime lost due to a partial CSMA/CA MAC slot before the start of an \emph{on} period (when the scheduled transmitter starts transmitting).
%
\end{theorem}
\begin{proof}
Let $z = \bar{T}_{{\rm off}} - c_1$, $\tilde{z}:=\log z$, $\tilde{s}_{{\rm csma},j} :=\log {s}_{{\rm csma},j}$ and $\tilde{s}_{{\rm txA}}:=\log {s}_{{\rm txA}}$.  Then,
\begin{align}
\tilde{s}_{{\rm csma},j} &= \log s_j \frac{\bar{T}_{{\rm off}} - c_1}{T_{\rm on}+\bar{T}_{{\rm off}}} 
= \log s_j +\tilde{z}-\log(T_{\rm on}+c_1+e^{\tilde{z}}),\notag
\end{align}
and
\small
\begin{align}
\tilde{s}_{{\rm txA}}&=\log r\frac{T_{\rm on}-c_2}{T_{\rm on}+\bar{T}_{{\rm off}}} 
=\log (r(T_{\rm on}-c_2)) -\log(T_{\rm on}+c_1+e^{\tilde{z}}).\notag
\end{align}
\normalsize
It can be verified (by inspection of the second derivative) that $\log(T_{\rm on}+c_1+e^{\tilde{z}})$ is convex in $\tilde{z}$ when $T_{\rm on}+c_1\ge 0$.  Hence, putting the network constraints in standard form,

\small
\begin{align}
\tilde{s}_{{\rm csma},j} -\log s_j -\tilde{z}+\log(T_{\rm on}+c_1+e^{\tilde{z}}) &\le 0,\ j=1,\dots,n\\
\tilde{s}_{{\rm txA}}-\log q +\log(T_{\rm on}+c_1+e^{\tilde{z}})&\le 0,
\end{align}
\normalsize
where $q:=r(T_{\rm on}-c_2)$, it can be seen that they are convex in decision variables $\tilde{s}_{{\rm csma},j}$, $\tilde{s}_{{\rm txA}}$ and $\tilde{z}$.

The proportional fair rate allocation for the scheduled network is the solution to the following utility optimisation,

\small
\begin{align*}
&\max_{\tilde{s}_{{\rm csma},j},\tilde{s}_{{\rm txA}},\tilde{z}} \tilde{s}_{{\rm txA}}+\sum_{j=1}^n\tilde{s}_{{\rm csma},j}\\
s.t.\quad&\tilde{s}_{{\rm csma},j} -\log s_j -\tilde{z}+\log(T_{\rm on}+c_1+e^{\tilde{z}}) \le 0,\ j=1,\dots,n\\
&\tilde{s}_{{\rm txA}}-\log q +\log(T_{\rm on}+c_1+e^{\tilde{z}})\le 0.
\end{align*}
\normalsize
The optmisation is convex and satisfies the Slater condition, hence strong duality holds.  The Lagrangian is,
\begin{align*}
L=&-\tilde{s}_{{\rm txA}}-\sum_{j=1}^n\tilde{s}_{{\rm csma},j} \\
&+\theta(\tilde{s}_{{\rm txA}}-\log q +\log(T_{\rm on}+c_1+e^{\tilde{z}}))\\
&+ \sum_{j=1}^n\lambda_j(\tilde{s}_{{\rm csma},j} -\log s_j -\tilde{z}+\log(T_{\rm on}+c_1+e^{\tilde{z}})).
\end{align*}
The main KKT conditions are
\begin{align}
&-1+\theta = 0, \quad
-1+\lambda_j = 0\,\ j=1,\dots,n \nonumber\\
&(\theta + \sum_{j=1}^n\lambda_j)\frac{e^{\tilde{z}}}{T_{\rm on}+c_1+e^{\tilde{z}}}-\sum_{j=1}^n\lambda_j=0.\nonumber
\end{align}
Thus, at an optimum $\theta=1$, $\lambda_j=1$, $j=1,\dots,n$ and
\begin{align}
\frac{e^{\tilde{z}}}{T_{\rm on}+c_1+e^{\tilde{z}}}=\frac{n}{n+1} \label{eq:toff}.
\end{align}
It can be verified (by inspection of the first derivative) that the LHS is monotonically increasing in $\tilde{z}$ and so a unique solution $\tilde{z}$ exists satisfying (\ref{eq:toff}).   Letting $\tilde{z}^*$ denote this solution, the proportional fair $\bar{T}_{{\rm off}}$ value is given by 
$\bar{T}_{{\rm off}}^* = e^{\tilde{z}^*}+c_1$.
The channel time fraction available for full CSMA/CA MAC slots is
\begin{align}
\frac{\bar{T}_{{\rm off}}^* -c_1}{T_{\rm on}+\bar{T}_{{\rm off}}^*} 
&= \frac{e^{\tilde{z}^*}}{T_{\rm on}+c_1+e^{\tilde{z}^*}} 
= \frac{n}{n+1}
\end{align}
and the channel time fraction used by the scheduled network is
\begin{align}
1-\frac{\bar{T}_{{\rm off}}^* -c_1}{T_{\rm on}+\bar{T}_{{\rm off}}^*} 
&=\frac{T_{\rm on}+c_1}{T_{\rm on}+\bar{T}_{{\rm off}}^*} 
=\frac{1}{n+1}.
\end{align}
\end{proof}

\begin{figure*}[hhht!] 
\centering
\subfigure[$n=1,T_{\rm on}=10$~ms]{\includegraphics[width=0.65\columnwidth]{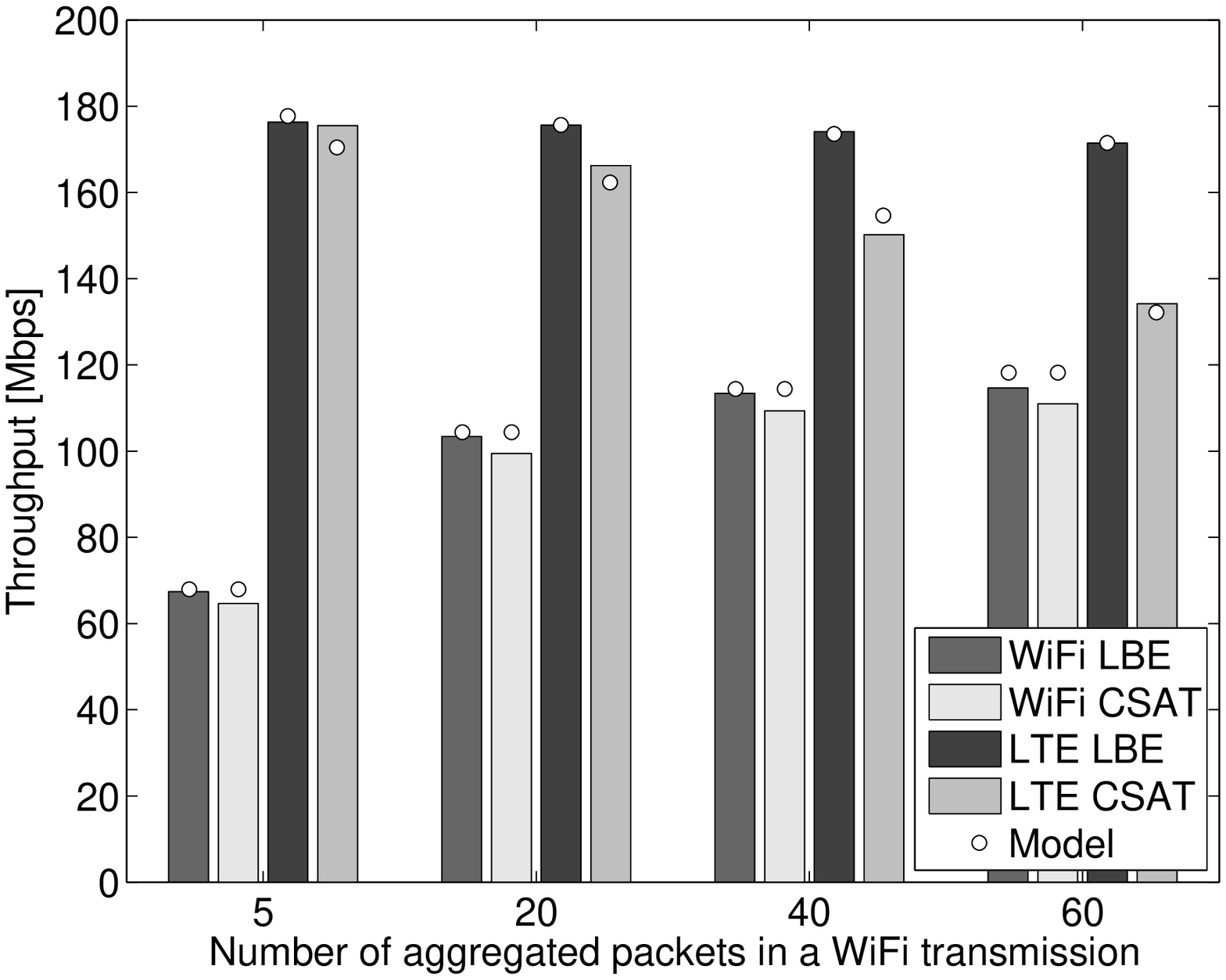}}
\subfigure[$n=3,T_{\rm on}=10$~ms]{\includegraphics[width=0.65\columnwidth]{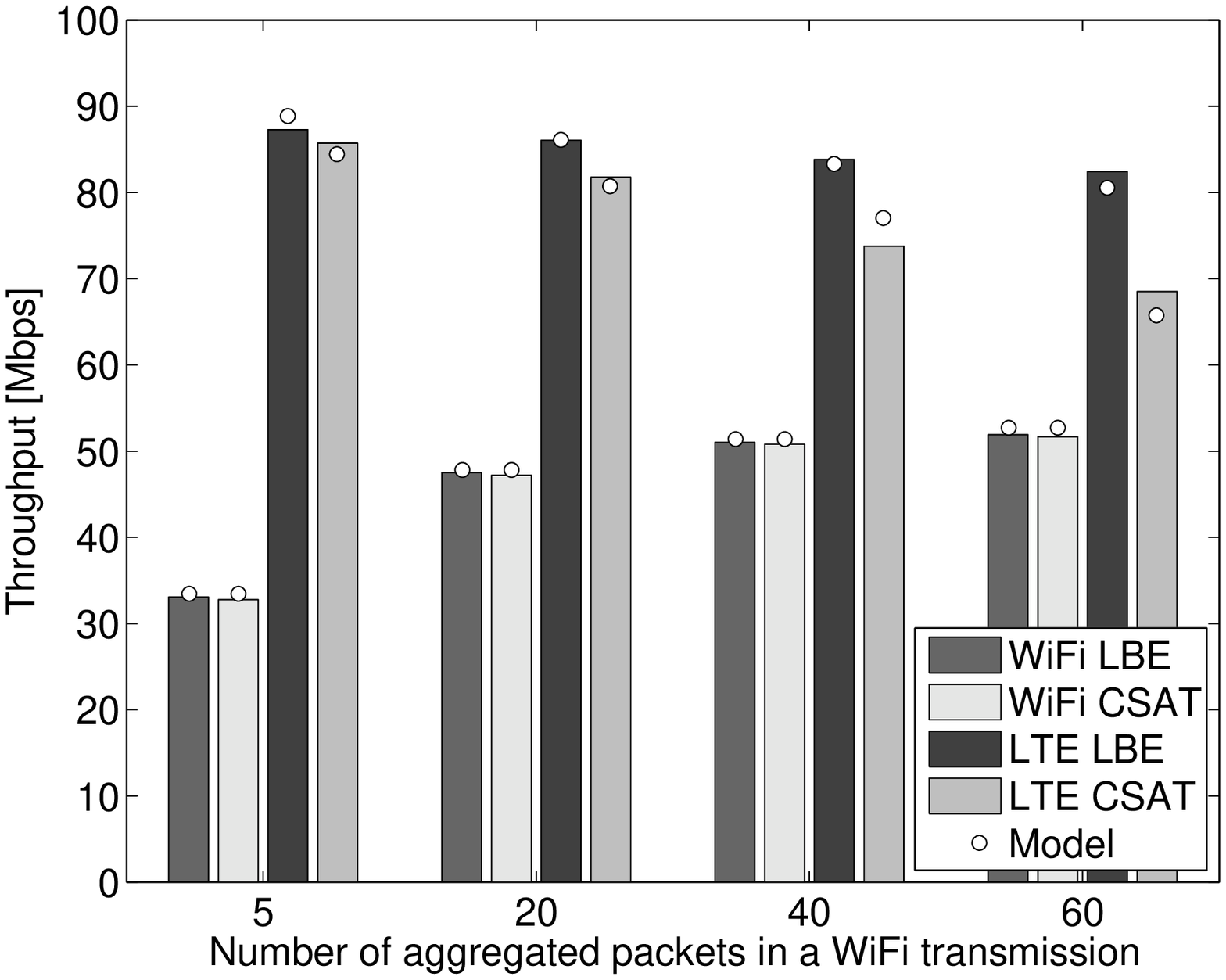}}
\subfigure[$n=9,T_{\rm on}=10$~ms]{\includegraphics[width=0.65\columnwidth]{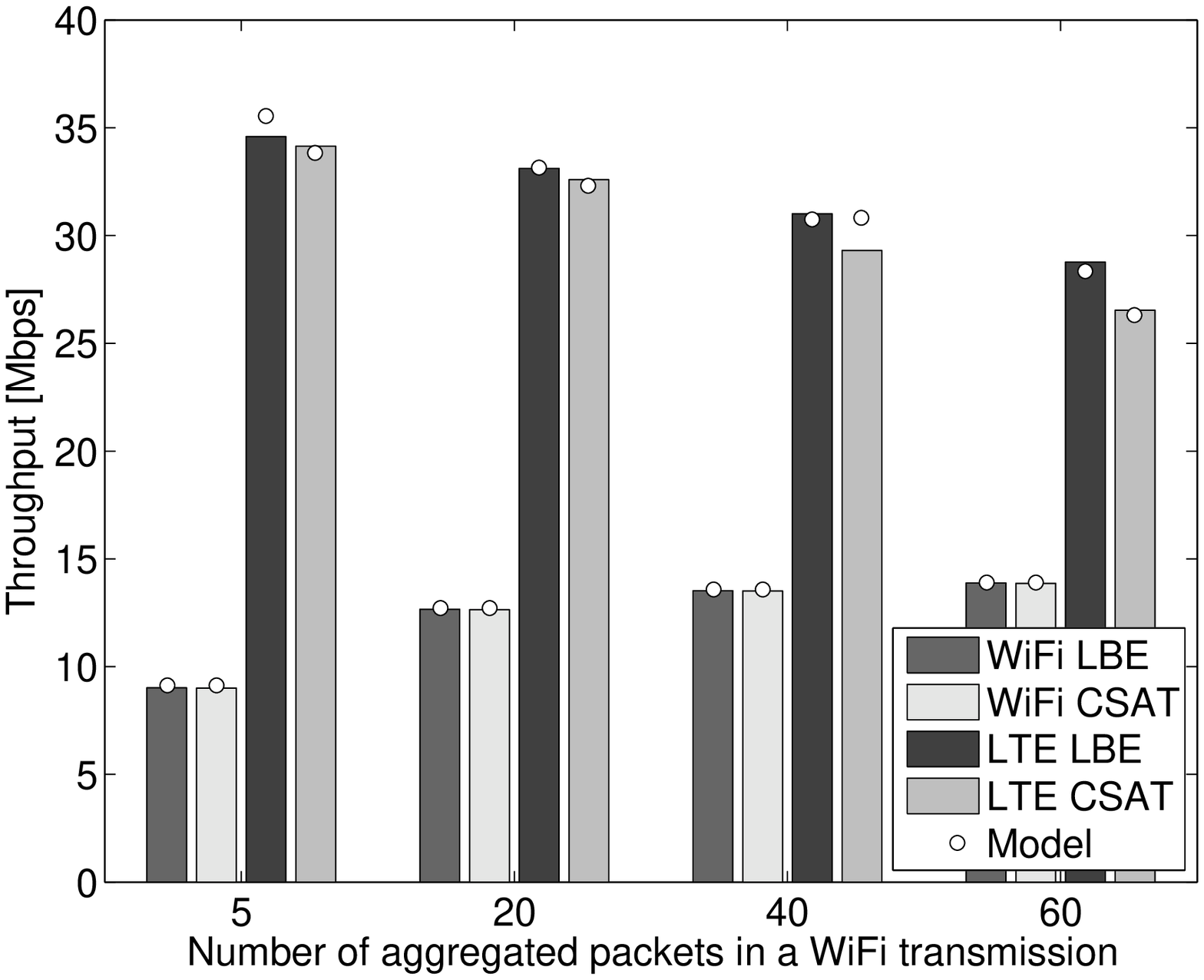}}\\
\subfigure[$n=1,T_{\rm on}=50$~ms]{\includegraphics[width=0.65\columnwidth]{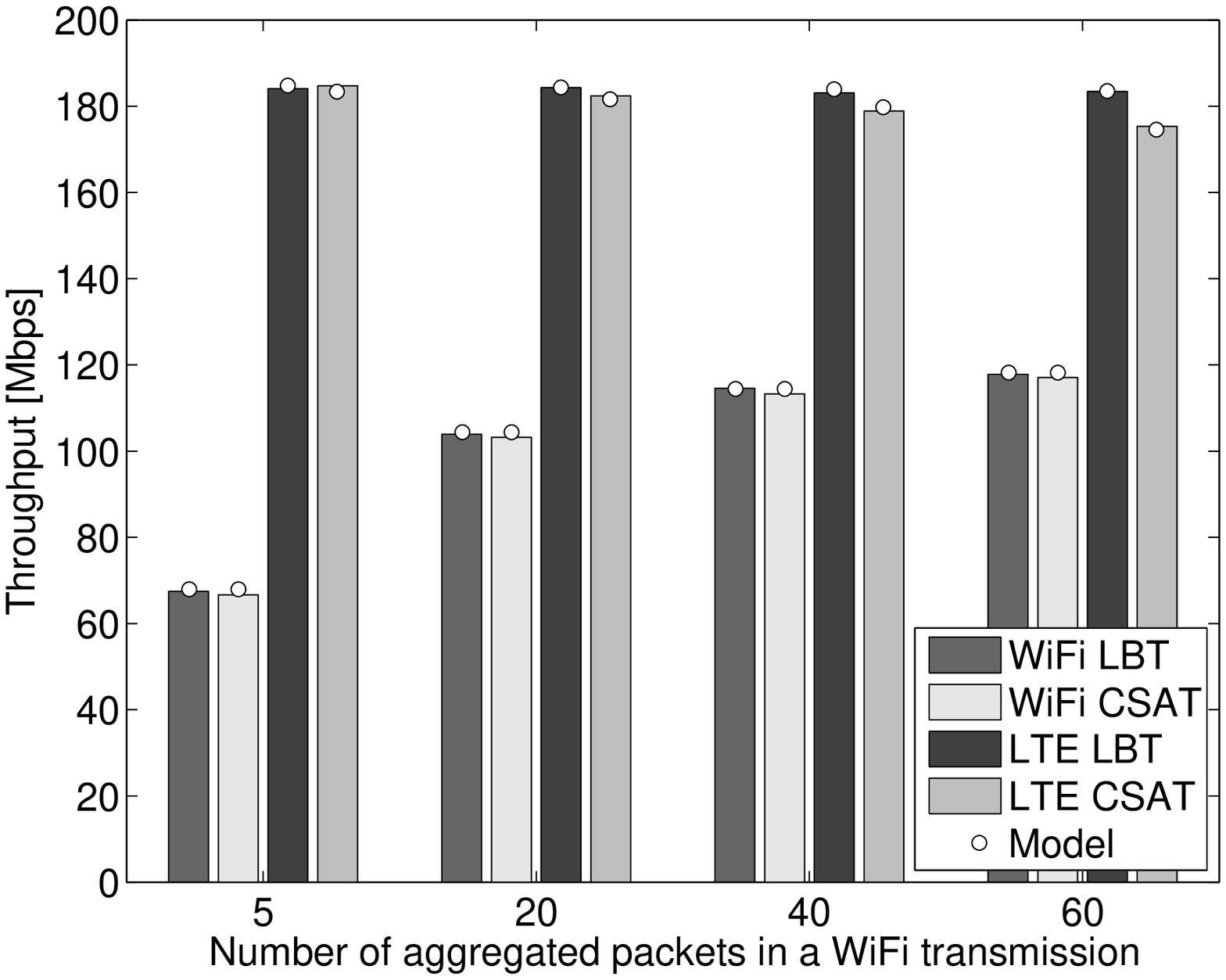}}
\subfigure[$n=3,T_{\rm on}=50$~ms]{\includegraphics[width=0.65\columnwidth]{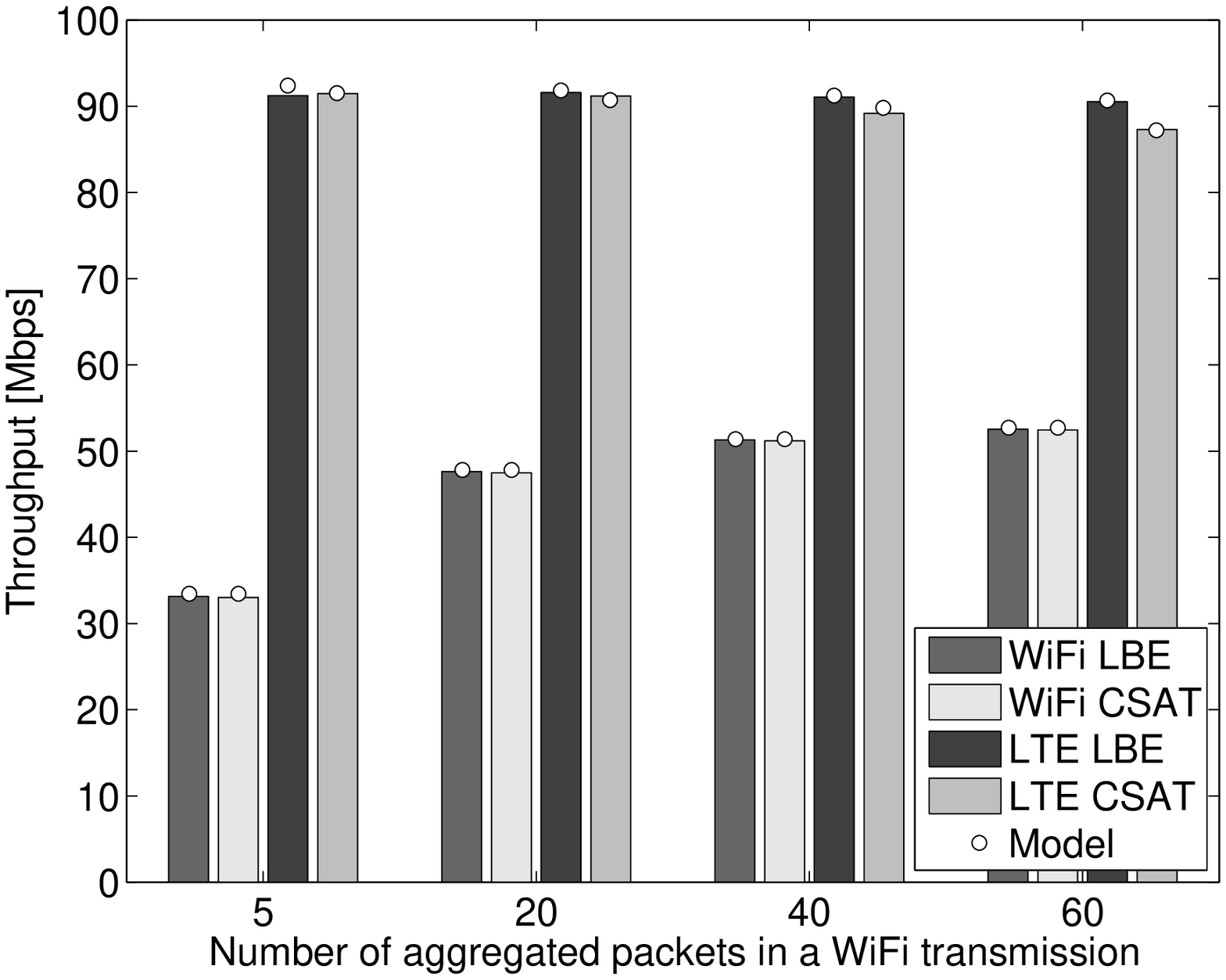}}
\subfigure[$n=9,T_{\rm on}=50$~ms]{\includegraphics[width=0.65\columnwidth]{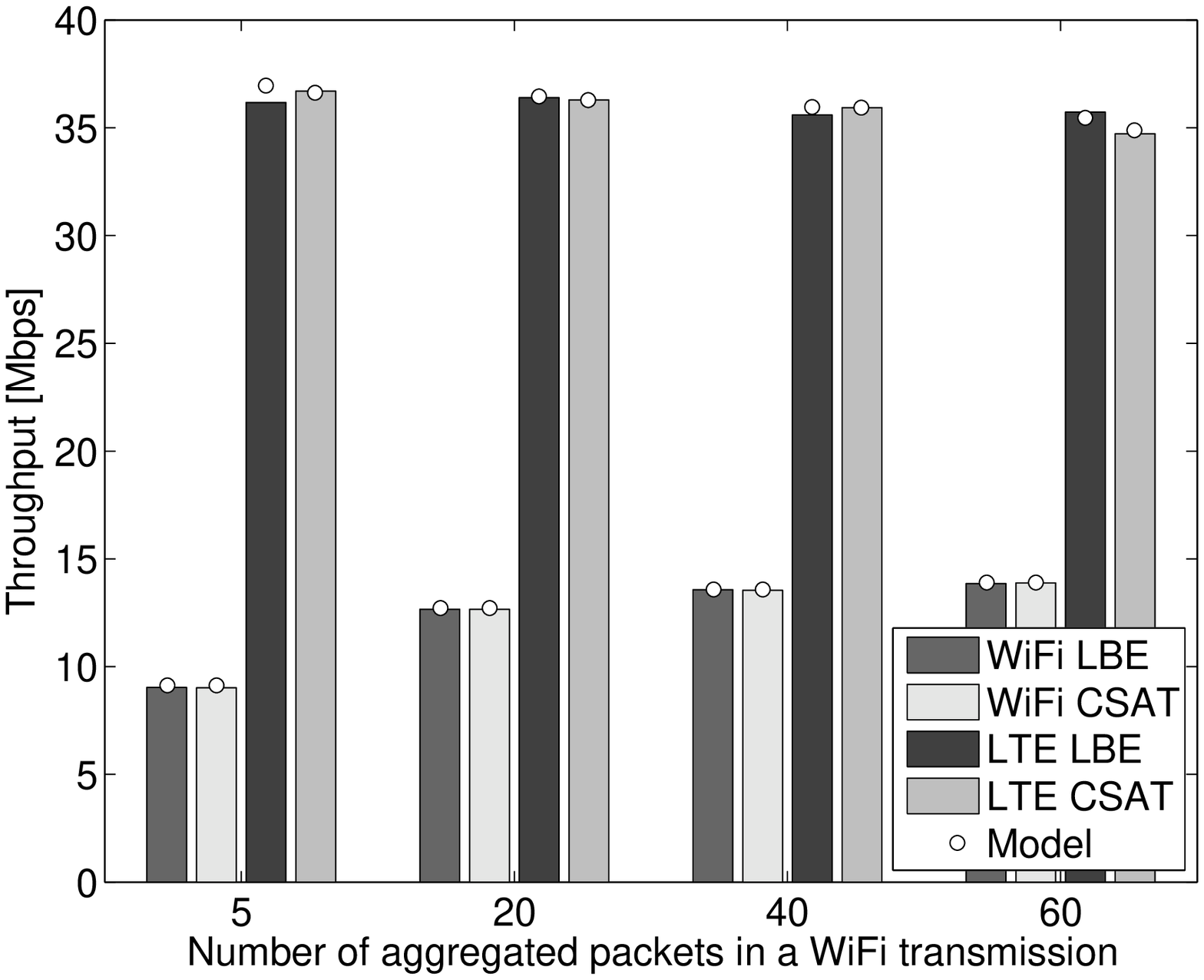}}
\caption{Proportional fair throughput allocation for different configurations of $n$ and $T_{\rm on}$ while varying $n_{\rm agg}$ (effectively changing the packet size of WiFi transmissions). Simulation results are averages of $100$ simulation runs with $50$ s time horizon.}
\label{fig:throughput_prop_fair}
\end{figure*}

\subsection{Discussion}\label{sec:discussion}
The scheduled transmission time $T_{\rm on}$ includes the reservation signal and/or the airtime due to collisions between scheduled and CSMA/CA transmitters.  Hence, $T_{\rm on}+c_1$ is the time spent transmitting plus the time spent on collisions, reservation signals and partial CSMA/CA MAC slots.  Letting $(\bar{T}_{{\rm off}}-c_1)/n$ denote the airtime \emph{allocated} to a CSMA/CA station, including idle time and collisions with other CSMA/CA nodes, then Theorem~\ref{tm:one} tells us that the proportional fair rate allocation equalises this airtime and $T_{\rm on}+c_1$.   

That is, the airtime allocated to the scheduled network is the same as the total channel time effectively used by the CSMA/CA network divided among the number of CSMA/CA transmitters.   This seems quite intuitive and is similar to previous proportional fair analysis for WiFi-only settings~\cite{checco2011proportional,2012_garcia-saavedra_infocom_ados}; the most interesting point here is that the proportional fair allocation assigns the cost of heterogeneity, i.e. the airtime cost of a collision between transmitter $A$ and the CSMA/CA network and of any reservation signals, to the scheduled network. On the other hand, the inefficiency of the random access mechanism (idle periods and collisions among CSMA/CA nodes) is accounted for in the total effective channel time of the CSMA/CA network.  Note that one immediate consequence of the cost of heterogeneity being accounted for in the airtime allocated to the scheduled network is that both the \emph{Preemptive} and \emph{Opportunistic} approaches, when 
configured for a proportional fair rate allocation, result in the same throughput for the CSMA/CA network.

The extension of this analysis to allow multiple users in the same scheduled network (i.e. where all of these users are synchronised to the same set of transmit slots) is straightforward, and in this case the airtime allocated to a user associated to the scheduled network is the same as that \emph{allocated} to a CSMA/CA station (again accounting for the cost of heterogeneity within the scheduled airtime). The case of multiple scheduled transmitters belonging to different networks will be discussed later in Section~\ref{sec:scope}.

%


\subsection{Example: Unlicensed LTE and 802.11}\label{sec:exampleLTEU}

We revisit here the example in Section~\ref{sec:ex} to illustrate the proportional fair allocation.  
As noted before, the \emph{Preemptive} approach can be identified with the LTE CSAT approach and the \emph{Opportunistic} approach with LBT/LBE.   We use the same MAC parameters as in Section~\ref{sec:ex}.  Additionally, for the LTE network throughput calculation we assume that the Control Format Indicator (CFI) is equal to 0 (i.e. we assume that the control information is sent through the licensed interface, which is in line with current 3GPP and LTE-U Forum discussions).    To obtain the proportional fair allocation our throughput model is applied with $\E[M]=\sigma p_{\rm e}+(T_{\rm b}+{\rm{DIFS}})(1-p_{\rm e})$ and in the \emph{Preemptive} case $p_{{\rm txA}} = \frac{p_{\rm s} T_{\rm b}+ p_{\rm c} T_{\rm fra} }{\E[M]}$.  This accounts for the interframe spaces defined in 802.11 as well as for $\Delta$ now being different in the case of a successful transmission versus a collision.

\subsubsection{Cost of Heterogeneity}

Fig.~\ref{fig:throughput_prop_fair} shows the WiFi and LTE proportional fair throughputs when using CSAT and LBE.   Results are shown both for detailed packet-level simulations and for the throughput model presented in Section~\ref{sec:appendix_model}.   These show the impact of varying $T_{\rm on}$, $n$ and the number of aggregated packets in a WiFi transmission $n_{\rm agg}$ (effectively changing the packet size).  

It can be seen that the WiFi throughput is essentially the same when using either CSAT and LBE for all configurations.   In contrast, however, the LTE throughput varies depending on the coexistence mechanism used and the network conditions.   For example, we can observe a considerable decrease in throughput when CSAT is used for $T_{\rm on}=10$~ms and larger WiFi packet sizes (see Figs.~\ref{fig:throughput_prop_fair}a-c). Further inspection reveals that the reason for this is the increased collision probability between LTE and WiFi transmissions when using CSAT compared to LBE.  

Collisions between LTE and WiFi are part of the cost of heterogeneity discussed in Section~\ref{sec:preliminaries-cost}, and in a proportional fair rate allocation are accounted for in the LTE channel airtime (by Theorem~\ref{tm:one}).   This cost of heterogeneity can be reduced by increasing the duration $T_{\rm on}$ of each LTE transmission.  For example, we can observe in Figs.~\ref{fig:throughput_prop_fair}d-f that both the CSAT and LBE schemes provide similar LTE throughput for $T_{\rm on}=50$~ms. However, the duration of LTE transmissions will tend to increase the delay experienced by WiFi and we consider this in more detail next.

\subsubsection{Throughput vs Delay}

Although increasing the duration of the LTE transmissions improves LTE throughput and reduces the cost of heterogeneity, it may also increase the delay experienced by WiFi transmissions when WiFi stations defer their transmissions while LTE transmissions are ongoing.   We investigate the distribution of the MAC access delay of WiFi packets when LTE uses CSAT and LBE in order to assess the trade-off between LTE throughput and WiFi delay.  

Fig.~\ref{fig:cdf_delay} shows the measured CDF of the WiFi MAC access delay when $n=1$, $n_{\rm agg}=60$ packets and for $T_{\rm on}=10$~ms and $T_{\rm on}=50$~ms. It can be seen that for a given value of $T_{\rm on}$, the distribution of the WiFi delay is similar for both CSAT and LBE. We can also see that increasing $T_{\rm on}$ causes longer delays for a fraction of the WiFi packets (namely, those whose transmisison has been deferred while an LTE transmission is in progress). Note that increasing $T_{\rm on}$ while maintaining the proportional fair configuration also causes the LTE network to access the channel less often, that is $T_{\rm off}$ also increases correspondingly. The consequence of this is that a higher percentage of the WiFi packets can access the channel during $T_{\rm off}$, experiencing short delays and so the \emph{mean} WiFi packet delay actually falls as the LTE $T_{\rm on}$ increases.   However, a fraction of WiFi packets experience long delays.  For example, for $T_{\rm on}=10$~ms 
it can be seen from Fig.~\ref{fig:cdf_delay} that around $73\%$ of the WiFi transmissions observe short delays, while for $T_{\rm on}=50$~ms, this percentage increases to $\sim 94\%$.

\begin{figure}[hht!] 
\centering
\includegraphics[width=0.70\columnwidth]{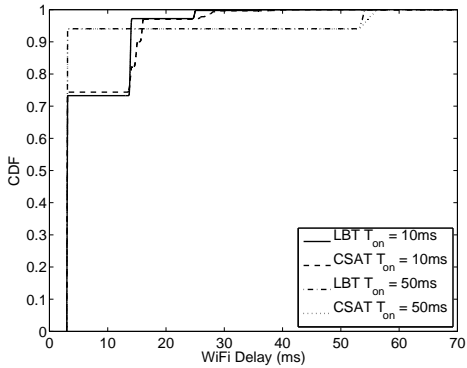}
\caption{CDF of delay for WiFi nodes with $n=1$, $n_{\rm agg}=60$ packets and for different $T_{\rm on}$ values. Results obtained from $100$ simulation runs with $50$~s time horizon.}
\label{fig:cdf_delay}
\end{figure}






%% file: unsat.tex
\section{Networks With a Mixture of Unsaturated and Saturated Stations}\label{sec:unsat}

Fair coexistence is only relevant when one or more stations are saturated (are persistently backlogged and so always have a packet to send). Otherwise, all stations can serve all of their offered load and there is no need to consider fairness in resource sharing.   The analysis in the preceding section assumes that all stations are saturated and in this section we extend consideration to situations where some CSMA/CA stations may be unsaturated\footnote{Extension to the scheduled network being unsaturated is discussed in Section~\ref{sec:scope}.}. 


\subsection{Utility Fair Optimisation}

When CSMA/CA stations are saturated their transmission attempt probability $\tau_j$ is fixed.  However, for unsaturated stations $\tau_j$ depends upon both the offered load and the network load (since the latter affects the mean MAC slot duration) as well as any buffer dynamics.  To extend our analysis to consider the case of unsaturated CSMA/CA stations we therefore make the following simplifying assumptions:  

\noindent 1. \emph{Small buffers}.  If the buffers are long, then during a $T_{\rm on}$ period unsaturated CSMA/CA stations may accumulate packets that are then transmitted during the $T_{\rm off}$ period. That is, the transmission activity of the scheduled network can therefore affect the transmission probability of the unsaturated CSMA/CA stations during the $T_{\rm off}$ period as it has an effect on buffer dynamics.   In order to avoid consideration of these buffer dynamics, which make the analysis much less tractable, we assume small buffers and neglect this effect.   The throughput model in Section~\ref{sec:appendix_model} then applies unchanged to the unsaturated case.

\noindent 2. \emph{${p}_{\rm e}$ constant at rate region boundary}.   Similarly to~\cite{subramanian2012convexity} we assume that at the boundary region the probability of a CSMA/CA MAC slot being empty can be considered constant i.e. ${p}_{\rm e}=\bar{p}_{\rm e}$ where $\bar{p}_{\rm e}$ is a fixed parameter.  As noted in~\cite{subramanian2012convexity}, this approximation is generally accurate and involves little loss.

%
With these assumptions we have, 


\small
\begin{align}
\tilde{s}_{{\rm csma},j} &= \tilde{x}_j + \log \bar{p}_{\rm e} + \log D_j - (\sigma \bar{p}_{\rm e}+\Delta(1-\bar{p}_{\rm e}))) \notag\\ & \qquad+\tilde{z} -\log(T_{\rm on}+c_1+e^{\tilde{z}}),\ j=1,\dots,n
\end{align}
\normalsize
where we have performed the change of variable $x_j=\tau_j/(1-\tau_j)$ resulting in $p_{succ,j}=x_j \bar{p}_{\rm e}$, and $\tilde{x}:=\log x$.

Rearranging the terms and putting the constraints in standard form, the proportional fair optimisation problem is:
\small
\begin{align}
&\max_{\stackrel{\tilde{s}_{{\rm csma},j},\tilde{s}_{{\rm txA}},}{\tilde{z} ,\tilde{x}_j}} \tilde{s}_{{\rm txA}}+\sum_{j=1}^n\tilde{s}_{{\rm csma},j}\notag\\
s.t.\quad & \tilde{s}_{{\rm txA}}-\log q +\log(T_{\rm on}+c_1+e^{\tilde{z}})\le 0,\\
& \tilde{s}_{{\rm csma},j} - \left(\tilde{x}_j + \log \bar{p}_{\rm e} + \log D_j - c_3(\tilde{z}) +\tilde{z} \right) \le 0, \label{eq:rate}\\
& \tilde{s}_{{\rm csma},j} -  \tilde{\bar{s}}_{{\rm csma},j}, \le 0,\ j=1,\dots,n \label{eq:load}\\
&\log \sum_{j=1}^n(1+e^{\tilde{{x}}_j }) \le -\log \bar{p}_{\rm e},\label{eq:idle}
\end{align}
\normalsize
where $c_3(\tilde{z}):= \log (\bar{p}_{\rm e} (\sigma - \Delta) + \Delta) + \log(T_{\rm on}+c_1+e^{\tilde{z}})$ and constraint (\ref{eq:load}) takes into account that the transmission probability is bounded by the offered load at a station (denoted by $\bar{s}_{{\rm csma},j}$).  Note that the CSMA/CA transmission attempt probabilities are now included as decision variables $x_j$, $j=1,\dots,n$ rather than being taken as constant as in the saturated case considered previously.

\subsection{Proportional Fair Allocation}

We begin by showing that at an optimum, constraint (\ref{eq:idle}) is tight.
\begin{lemma}\label{lem1}
Suppose $C\subset\{1,\cdots,n\}$, the set of CSMA/CA stations for which the optimal rate $\tilde{s}^*_{{\rm csma},j} < \tilde{\bar{s}}_{{\rm csma},j}$ is not empty.  Then at an optimum $\log \sum_{j=1}^n(1+e^{\tilde{x}^*_j }) = -\log \bar{p}_{\rm e}$ i.e.  $\frac{1}{\prod_{j=1}^n(1+x^*_j )} =\bar{p}_{\rm e} $.
\end{lemma}
\begin{proof}
We proceed by contradiction.  Suppose that $\tilde{x}^*_j$, $j=1,\cdots,n$ is an optimum and $\log \sum_{j=1}^n(1+e^{\tilde{x}^*_j }) < -\log \bar{p}_{\rm e}$.  We can therefore increase $\tilde x^*_j$ for one or more stations without violating constraint (\ref{eq:idle}).  Since for at least one station constraint (\ref{eq:load}) is loose, by increasing $\tilde x^*_j$ for that station then constraint (\ref{eq:rate}) becomes loose (since increasing $\tilde x^*_j$ increases the RHS of (\ref{eq:rate}) without changing the LHS).  Constraints (\ref{eq:idle}) and (\ref{eq:rate}) are now both loose for that station and so we can increase $\tilde{s}_{{\rm csma},j}$.  But increasing $\tilde{s}_{{\rm csma},j}$ improves the objective of the optimisation, contradicting the assumption that we are at an optimum.
\end{proof}

With Lemma~\ref{lem1} in hand we can now state the proportional fair rate allocation:
\begin{theorem}\label{th:two}
The proportional fair rate allocation assigns the following fraction of airtime to the CSMA/CA transmitters:
\begin{align}
\frac{\bar{T}_{{\rm off}}^* -c_1}{T_{\rm on}+\bar{T}_{{\rm off}}^*} =\frac{|C|+\sum_{j\notin C}\lambda^*_j}{1 + |C|+ \sum_{j\notin C}\lambda^*_j},
\end{align}
where $C$ is the set of saturated CSMA/CA stations (assumed to be non-empty) and $\lambda^*_j$ is the fraction of airtime used by CSMA/CA station $j$ relative to that used by a saturated CSMA/CA station (so $\lambda^*_j=1$ for saturated stations and $\lambda^*_j<1$ for unsaturated stations).
\end{theorem}
\begin{proof}
The optimisation is jointly convex in the decision variables and the Lagrangian is,
\small
\begin{align*}
L=&-\tilde{s}_{{\rm txA}}-\sum_{j=1}^n\tilde{s}_{{\rm csma},j} 
+\theta(\tilde{s}_{{\rm txA}}-\log q +\log(T_{\rm on}+c_1+e^{\tilde{z}}))\\
&+ \sum_{j=1}^n\lambda_j \left(\tilde{s}_{{\rm csma},j} - \left( \tilde{x}_j + \log \bar{p}_{\rm e} + \log D_j - c_3(\tilde{z}) +\tilde{z}\right) \right)\\
&+ \sum_{j=1}^n\mu_j \left(\tilde{s}_{{\rm csma},j} -  \tilde{\bar{s}}_{{\rm csma},j}\right)
+\gamma(\log \sum_{j=1}^n(1+e^{\tilde{{x}}_j }) +\log \bar{p}_{\rm e} ).
\end{align*}
\normalsize
The main KKT conditions are:

\small
\begin{align}
&-1+\theta^* = 0, \quad
-1+\mu^*_j+\lambda^*_j = 0,\ j=1,\dots,n \nonumber\\
&(\theta^* + \sum_{j=1}^n\lambda^*_j)\frac{e^{\tilde{z}^*}}{T_{\rm on}+c_1+e^{\tilde{z}^*}}-\sum_{j=1}^n\lambda^*_j=0,\nonumber\\
&-\lambda^*_j+\gamma^*\frac{e^{\tilde{x}^*_j}}{\sum_{j=1}^n(1+e^{\tilde{x}^*_j }) } = 0.\nonumber
\end{align}
\normalsize
Therefore at an optimum we have $\theta^* = 1$ and $\mu^*_j = 1-\lambda^*_j$. Letting $C\subset\{1,\cdots,n\}$ be the set of CSMA/CA stations for which the optimum rate $\tilde{s}^*_{{\rm csma},j} < \tilde{\bar{s}}_{{\rm csma},j}$ (the set of saturated stations), then by complementary slackness we have $\mu^*_j=0$, $j\in C$ and so $\lambda^*_j=1$, $j\in C$.  Combining the above with the KKT condition for $\tilde{z}$ we have that at an optimum,
\begin{align*}
&\frac{e^{\tilde{z}^*}}{T_{\rm on}+c_1+e^{\tilde{z}^*}} = \frac{\bar{T}_{{\rm off}}^* -c_1}{T_{\rm on}+\bar{T}_{{\rm off}}^*} =\frac{|C|+\sum_{j\notin C}\lambda^*_j}{1 + |C|+ \sum_{j\notin C}\lambda^*_j},
\end{align*}
 where $\lambda^*_j=\gamma^*\frac{e^{\tilde{x}^*_j}}{\sum_{j=1}^n(1+e^{\tilde{x}^*_j }) } = \gamma^*\frac{x_j^*}{\sum_{j=1}^n(1+x_j^*) }$.
 
It remains to obtain $\lambda^*_j$ for $j \notin C$ (for the set of unsaturated stations).  We proceed by noting that the airtime used by CSMA/CA station $j$ for successful transmissions is given by
\begin{align}
{T}_{{\rm csma},j} &= \frac{\frac{\tau_j}{1-\tau_j}p_{\rm e}\Delta}{\sigma p_{\rm e} + \Delta(1-p_e)} \frac{z}{T_{\rm on}+c_1+z},\label{eq:succairtime}
\end{align} 
where $p_{\rm e}:=\prod_{k=1}^n(1-\tau_k)=\frac{1}{\prod_{k=1}^n(1+e^{\tilde{{x}}_k })}$ is the idle probability. 
By Lemma~\ref{lem1} at an optimum
\begin{align}
{T}^*_{{\rm csma},j} &=x^*_j\bar{p}_{\rm e}\frac{\Delta}{c_4} \frac{z^*}{T_{\rm on}+c_1+z^*},
\end{align} 
with $c_4 = (\bar{p}_{\rm e} (\sigma - \Delta) + \Delta)$. From the KKT conditions it follows that $\lambda^*_j= x^*_j \gamma^* \bar{p}_{\rm e}$.  Hence,
\begin{align}
{T}^*_{{\rm csma},j} &=\lambda^*_j\frac{\Delta}{\gamma^* c_4} \frac{z^*}{T_{\rm on}+c_1+z^*}.
\end{align}
For a station $j\in C$ we know that $\lambda^*_j=1$ and so at an optimum ${T}^*_{{\rm csma},j} = {T}^*_{{\rm csma},C} :=\frac{\Delta}{\gamma c_4} \frac{z^*}{T_{\rm on}+c_1+z^*}$.  Observe that the optimal airtime ${T}^*_{{\rm csma},j}$ has the same value for all $j\in C$ since $\gamma$ does not depend on $j$ i.e. all stations unconstrained by offered load are allocated the same success airtime.  For the stations $j\notin C$ from the above analysis we have $\lambda^*_j = \frac{{T}^*_{{\rm csma},j}}{{T}^*_{{\rm csma},C}}<1$ while for stations $j\in C$ $\lambda^*_j = \frac{{T}^*_{{\rm csma},j}}{{T}^*_{{\rm csma},C}}=1$.   That is, $\lambda^*_j = \frac{{T}^*_{{\rm csma},j}}{{T}^*_{{\rm csma},C}}$ for all stations, and this is just the fraction of airtime used by CSMA/CA station $j$ relative to that used by a CSMA/CA station unconstrained by offered load.  
%
\end{proof}

\subsection{Discussion}
It can be seen that Theorem~\ref{th:two} reduces to Theorem~\ref{tm:one} when all of the CSMA/CA stations are saturated.  Further, Theorem~\ref{th:two} has an elegant channel time interpretation.  Namely, it states that the airtime allocated to the scheduled network is the same as the total channel time effectively used by the CSMA/CA network divided among an \emph{equivalent} number of stations. This \emph{equivalent} number of CSMA/CA stations is computed based on the proportion of success airtime of an unsaturated station compared to a saturated one.

\subsection{Example: Unlicensed LTE and 802.11}

We revisit here again the example in Section~\ref{sec:exampleLTEU} to illustrate the results obtained in the analysis above when there is a set of unsaturated WiFi stations. The same parameters as in the former example are used except that here in order to avoid consideration of the impact of the buffer dynamics, we assume no aggregation. To compute $\bar{p}_{\rm e}$ as well as $\tau_j,\ j=1,\dots,n$ for a given WiFi  configuration, we have used a standard unsaturated model~\cite{bellalta2005simple}, considering that the buffer size is now limited to 5 packets of 1500 bytes. 

Fig.~\ref{fig:throughput_prop_fair_unsat} shows for CSAT the proportion of allocated channel time and successful airtime for the cases where: \emph{i)} $n=3, |C|=2$ and, \emph{ii)} $n=9, |C|=5$. In the former, the unsaturated offered load equals 10 Mbps while in the latter it is set to 3 Mbps. In the figure, the \emph{Assigned Channel Time} proportion is equal to ${T}_{{\rm on}}/(T_{\rm on}+\bar{T}_{{\rm off}})$ for LTE. In the case of WiFi, we define the \emph{Assigned Channel Time} to be $(1-\bar{T}_{{\rm off}})/n_{\rm eq}$ for a saturated WiFi station and $\lambda^*_j(1-\bar{T}_{{\rm off}}^*)/n_{\rm eq}$ for an unsaturated one, whith $n_{\rm eq}=|C|+\sum_{j\notin C}\lambda^*_j$. This allows us to evaluate the channel time \emph{allocated} to a station, although that includes empty periods and collision airtimes for WiFi, which are in fact shared among all WiFi stations. The successful airtime is also depicted in Fig.~\ref{fig:throughput_prop_fair_unsat} and it corresponds to $({T}_{{\rm on}}-c1)/(T_{\rm on}+\bar{T}_{{\rm off}})$ for LTE and to (\ref{eq:succairtime}) for WiFi.

We can see in Fig.~\ref{fig:throughput_prop_fair_unsat} how the solution of the proportional fair optimisation problem assigns equal channel times to LTE and to a saturated WiFi station. We can also note the inefficiency of the WiFi channel access and how its cost is not shared equally among all WiFi stations but varies proportionally to the offered load.  

\begin{figure}
\centering
\subfigure[$n=3, |C|=2$]{\includegraphics[width=0.65\columnwidth]{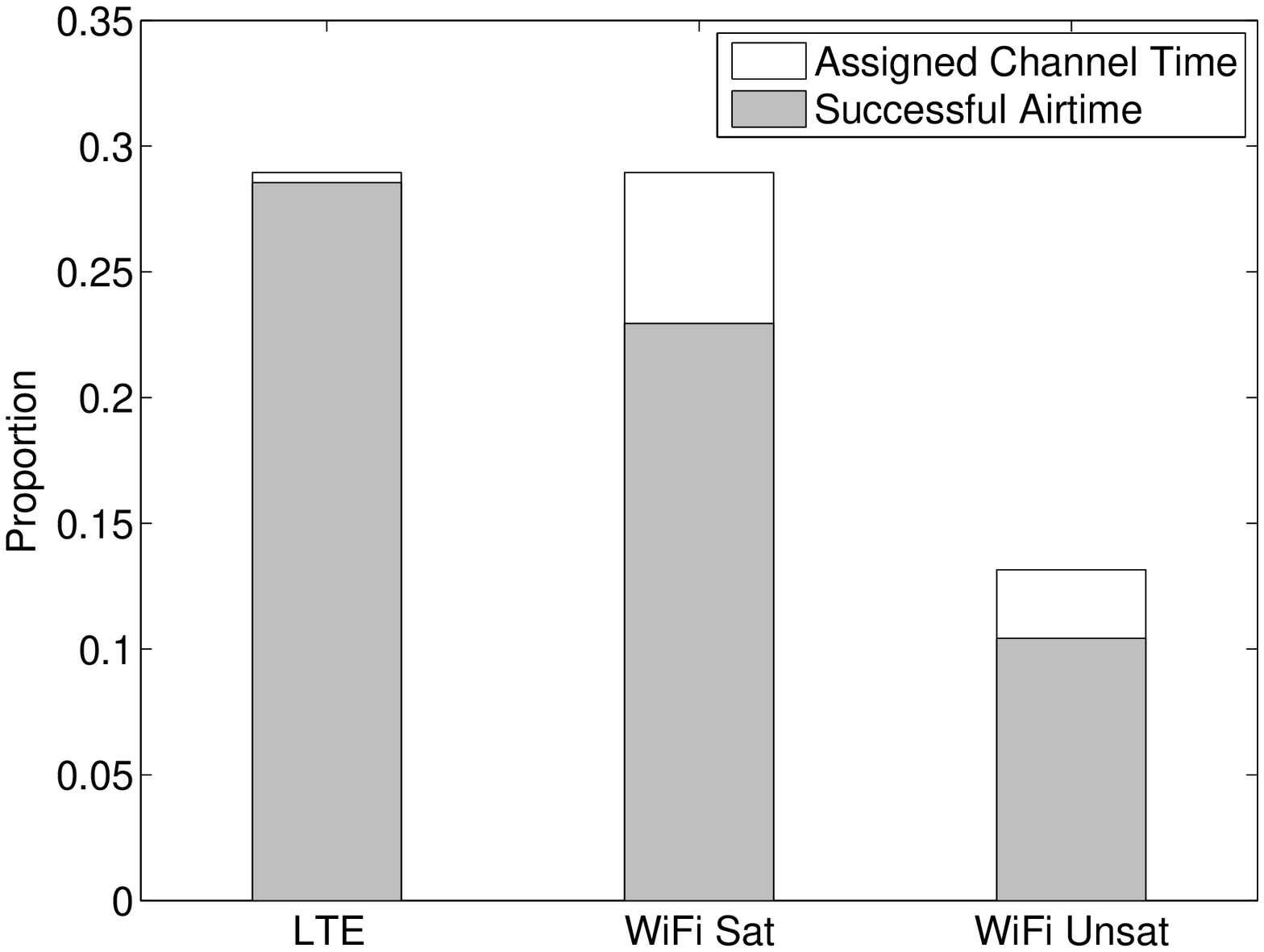}}
\subfigure[$n=9, |C|=5$]{\includegraphics[width=0.65\columnwidth]{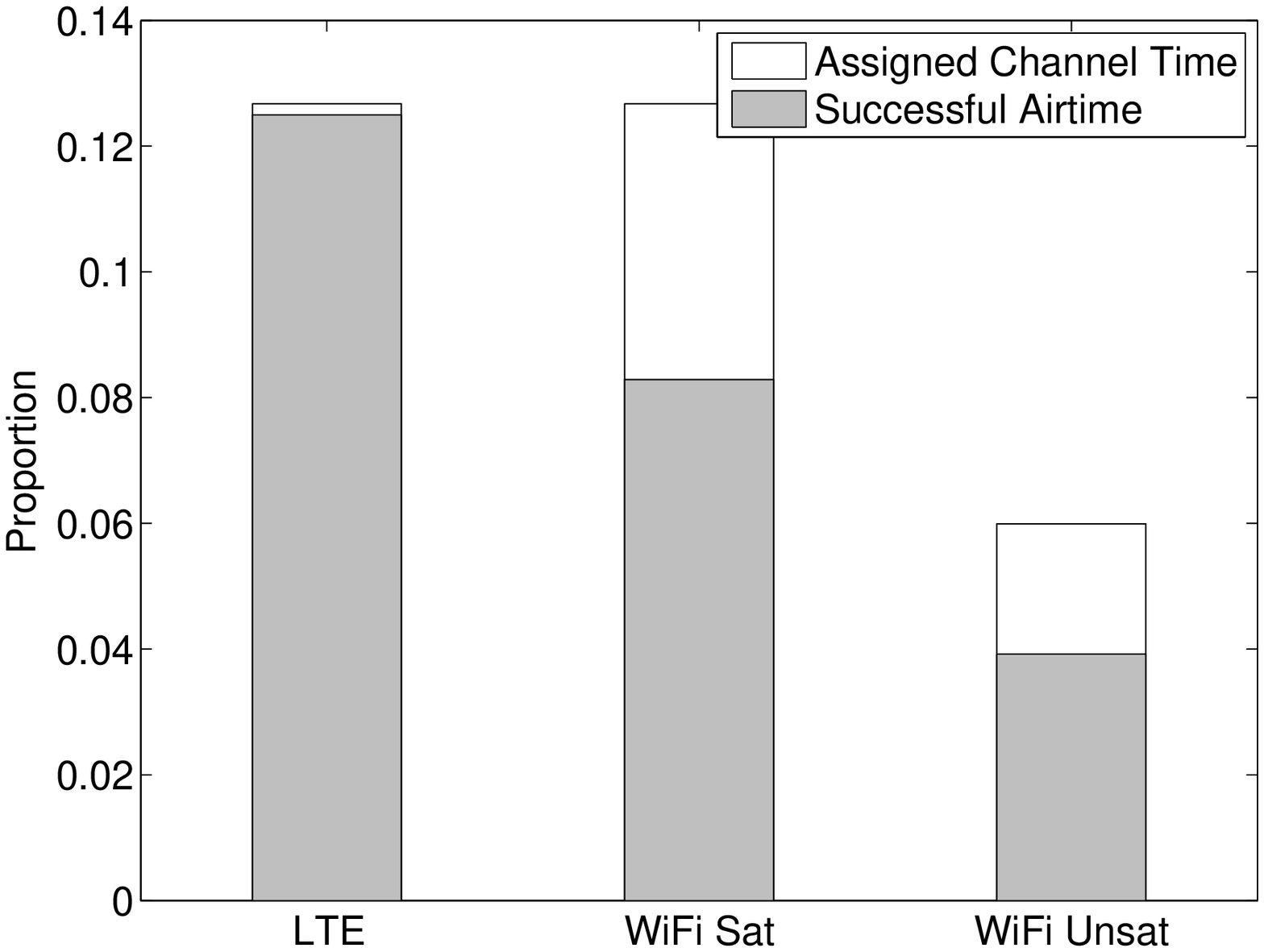}}
\caption{Resulting airtimes using the proportional fair allocation result for different configurations of $n$ and $C$.}
\label{fig:throughput_prop_fair_unsat}
\end{figure}

%% file: carriersense.tex
\section{Imperfect Inter-Technology Transmission Detection}\label{sec:carrier_sense}

As already noted, our focus here is on random access transmitters which use carrier sensing to define MAC slots.  In the foregoing analysis we have assumed that the random access transmitters can use their carrier sensing ability to also detect scheduled transmissions in fixed time slots and so defer their channel attempts while scheduled transmissions take place.   In this section we now relax this assumption and consider imperfect sensing of scheduled transmissions by the random access transmitters.

To help motivate this analysis we begin by giving a brief overview of the carrier sensing used by 802.11 transmitters and then present experimental measurements evaluating the effectiveness of this carrier sensing at detecting LTE transmissions in the unlicensed band.   Use of experimental measurements rather than simulations is important since not only is signal propagation indoors complex but also physical carrier sensing functionality is typically hardware-dependent.   






\subsection{Transmission Detection in 802.11}



The 802.11 standard mandates two types of detection of ongoing transmissions, namely virtual carrier sensing and physical carrier sensing.  Virtual carrier sensing operates at the MAC layer.  Transmitters set a duration field in the MAC header and receivers set a Network Allocation Vector (NAV) timer accordingly to mark the channel as busy for the duration requested in the transmitted frame.
Physical carrier sensing is carried out at the PHY layer and employs one or more of the following methods: (i) Energy Detection, which declares the channel to be busy when the received energy rises above a specified threshold, (ii) Weak Carrier Sensing, which detects the presence of OFDM transmissions, and (iii) Preamble Reception, which decodes the PLCP preamble to extract the \emph{Length} parameter which states the duration of the subsequent transmission.   

For detection of non-802.11 transmissions it is primarily physical carrier sensing using Energy Detection that is relevant since virtual carrier sensing and Preamble Reception are both 802.11-specific and Weak Carrier Sensing may also use 802.11-specific OFDM features.  In general, Energy Detection is the least sensitive form of carrier sensing as it makes use of energy measurements instead of decoded information and is therefore prone to false negatives unless the energy detection threshold is set sufficiently high.

\subsection{Testbed Hardware and Software Setup}

We constructed a small test-bed to assess the ability of WiFi devices to detect unlicensed-band LTE transmissions. 

\subsubsection{LTE SDR Transmitter} 
We used an Ettus USRP~B210 board, connected via an USB~3.0 interface to a standard PC (Intel Core i7) running Linux Ubuntu Trusty, with the \texttt{uhd\_driver} and version 1.0.0 of \texttt{srsLTE},\footnote{\url{https://github.com/srsLTE/srsLTE}} which is a free, open-source LTE library for implementing both an UE and an eNodeB.  The USRP board acts as eNodeB, configured to use 100~physical resource blocks (i.e., 20~MHz bandwidth) in the 5~GHz band, MCS index 0 and implementing a periodic duty-cycle channel access scheme, similar to the proposed CSAT coexistence mechanism~\cite{sadek2015extending}. This was achieved by modifying the \texttt{srsLTE} example program \texttt{src/examples/pdsch\_enodeb.c} to fix an active interval during which data is transmitted, followed by a silent period of random duration. The average total duration of the active plus silent periods is set equal to 100~ms, and by varying the mean duration of the silent period we effectively vary the \emph{duty cycle} of the LTE link.  The 
LTE transmission power is also varied.

\subsubsection{WiFi SDR Transmitter} 
To provide a baseline comparison, in our experiments we also operate the USRP board as an IEEE 802.11a transceiver\footnote{\url{https://github.com/bastibl/gr-ieee802-11}}.  We generate a similar \emph{on-off} WiFi transmission pattern as the one used for the LTE transmissions, in the same frequency channel and using the 6~Mbps MCS (which is the closest 802.11-compliant modulation and coding scheme to that of the LTE transmissions).   The \emph{on-off} WiFi transmissions are generated by transmitting WiFi packets in bursts during the \emph{on} periods (with no idle time between packets, i.e. no DIFS and no random backoff), while remaining silent during the random \emph{off} periods. 

\subsubsection{WiFi Receiver} 
WiFi channel sensing is performed by a Soekris net6501-70 device equipped with an Atheros~AR9390-based 802.11a card, running Linux Ubuntu (kernel 3.13) and the \texttt{ath9k} wireless driver.  We take advantage of this driver's monitoring capabilities to obtain the status of the medium as detected by the wireless card. This is achieved by leveraging \texttt{ath9k}'s 32-bit register counters \texttt{AR\_CCCNT}, \texttt{AR\_RCCNT}, and \texttt{AR\_RFCNT}, which count, respectively, the total number of cycles elapsed (with a 44~MHz resolution), the ones where the medium is marked as busy and the ones where there is an ongoing frame. With this data, we are able to measure the ``CCA state'' of the wireless medium (a similar approach is also used by \texttt{RegMon}~\cite{regmon}).

\begin{figure}
\centering
\subfigure[CSAT-based]{\includegraphics[width=\columnwidth]{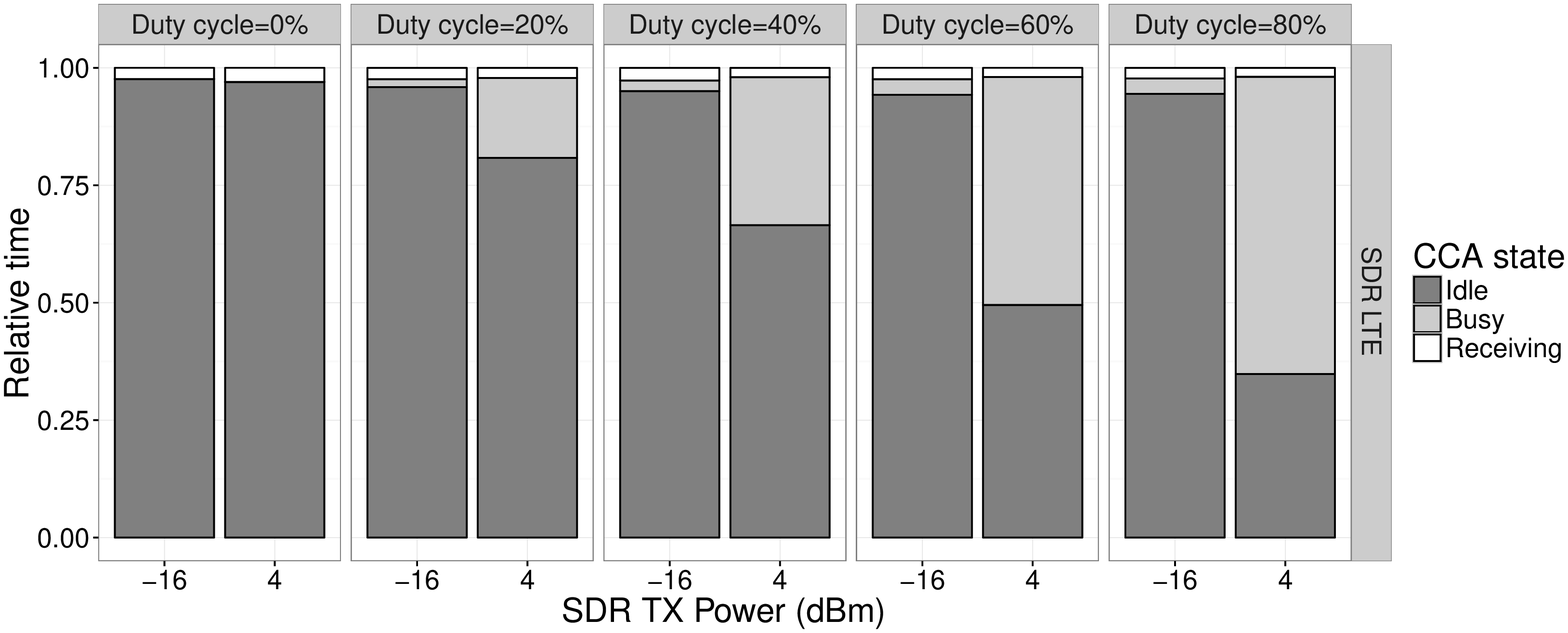}\label{fig:cca_exp_csat}}
\subfigure[802.11a]{\includegraphics[width=\columnwidth]{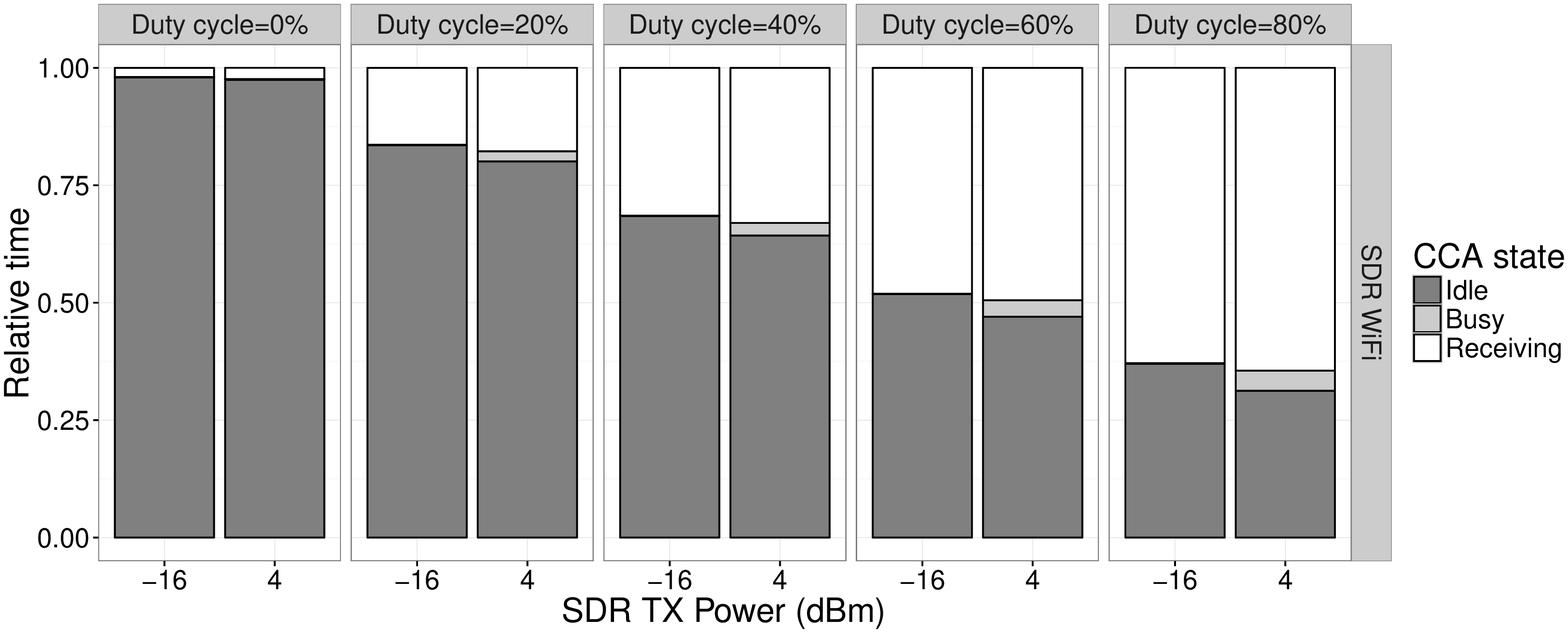}\label{fig:cca_exp_wifi}}
\caption{WiFi Atheros CCA states with an SDR CSAT-based implementation and  an SDR 802.11a implementation.}
\label{fig:cca_exp}
\end{figure}

\subsection{Experimental Measurements}

Fig.~\ref{fig:cca_exp_csat} plots the measured ``CCA state'' at the WiFi receiver as the configuration (duty cycle and transmit power) of the LTE transmitter is varied.   We start by considering the case with -16~dBm transmit power.  We note that, given the small size of the testbed, for this configuration the signal quality was very good (we confirmed this using another Ettus board, configured as an unlicensed LTE UE).  Nevertheless, it can be seen that the WiFi card consistently marks the channel as roughly 90\% idle even when the LTE duty cycle is 80\%.   When the LTE transmit power is increased substantially to 4~dBm, it can be seen that the situation changes and the WiFi card now marks the channel as becoming increasingly busy as the LTE duty cycle is increased. These measurements therefore show that WiFi carrier sensing fails to work at lower LTE transmit powers.

For comparison Fig.~\ref{fig:cca_exp_wifi} shows the corresponding measurements when the Ettus board is configured as a WiFi transmitter.   It can be seen that the WiFi card correctly detects the medium as being occupied by 802.11 transmissions for both values of the transmit power, closely following the duty cycle. 

These measurements demonstrate that the Energy Detection physical carrier sensing used by WiFi to detect LTE transmissions can be much less sensitive than the carrier sensing used by a WiFi station to detect WiFi transmissions.   

\begin{figure*}[hhht!] 
\centering
\subfigure[$n=1,T_{\rm on}=10$ms]{\includegraphics[width=0.65\columnwidth]{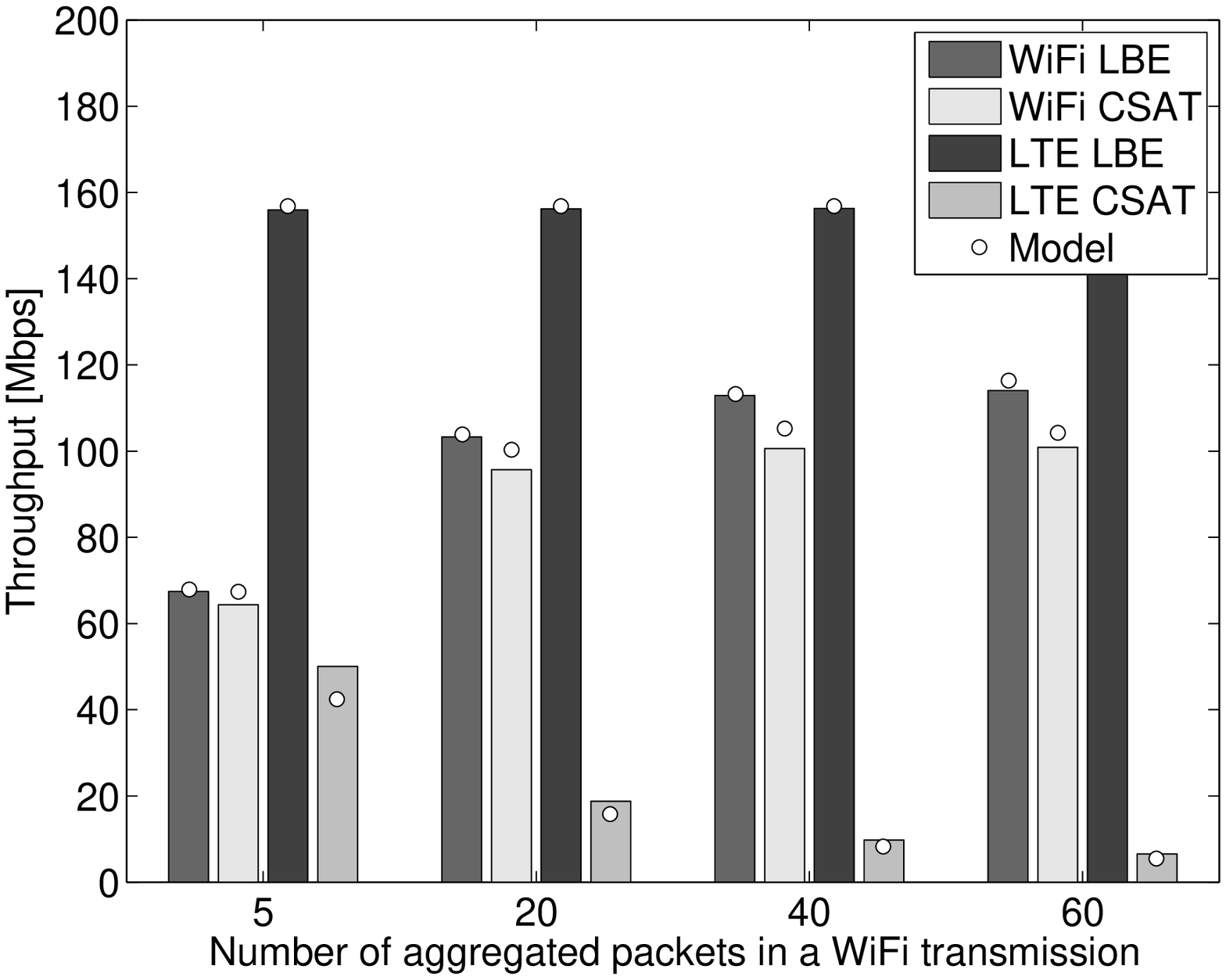}}
\subfigure[$n=3,T_{\rm on}=10$ms]{\includegraphics[width=0.65\columnwidth]{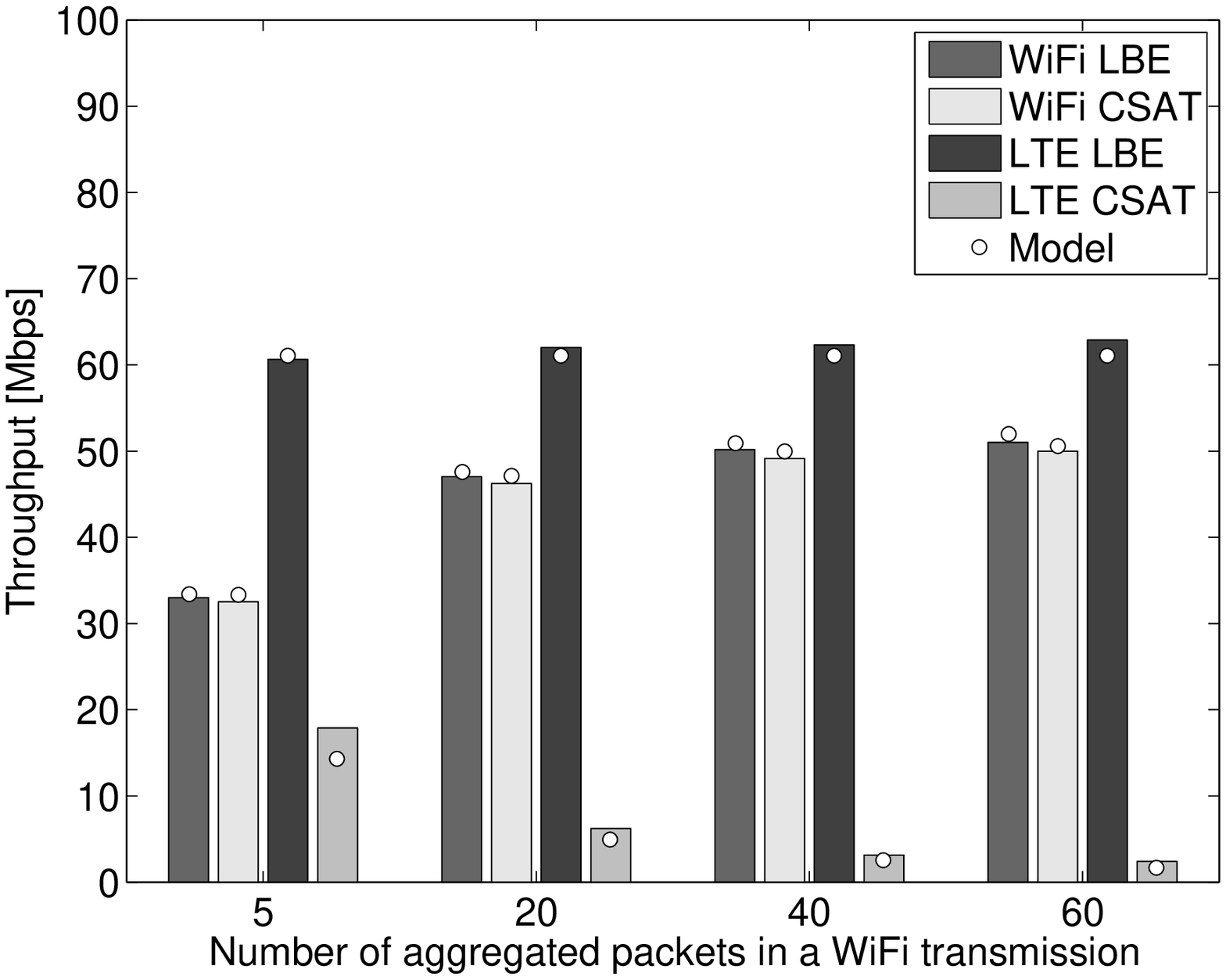}}
\subfigure[$n=9,T_{\rm on}=10$ms]{\includegraphics[width=0.65\columnwidth]{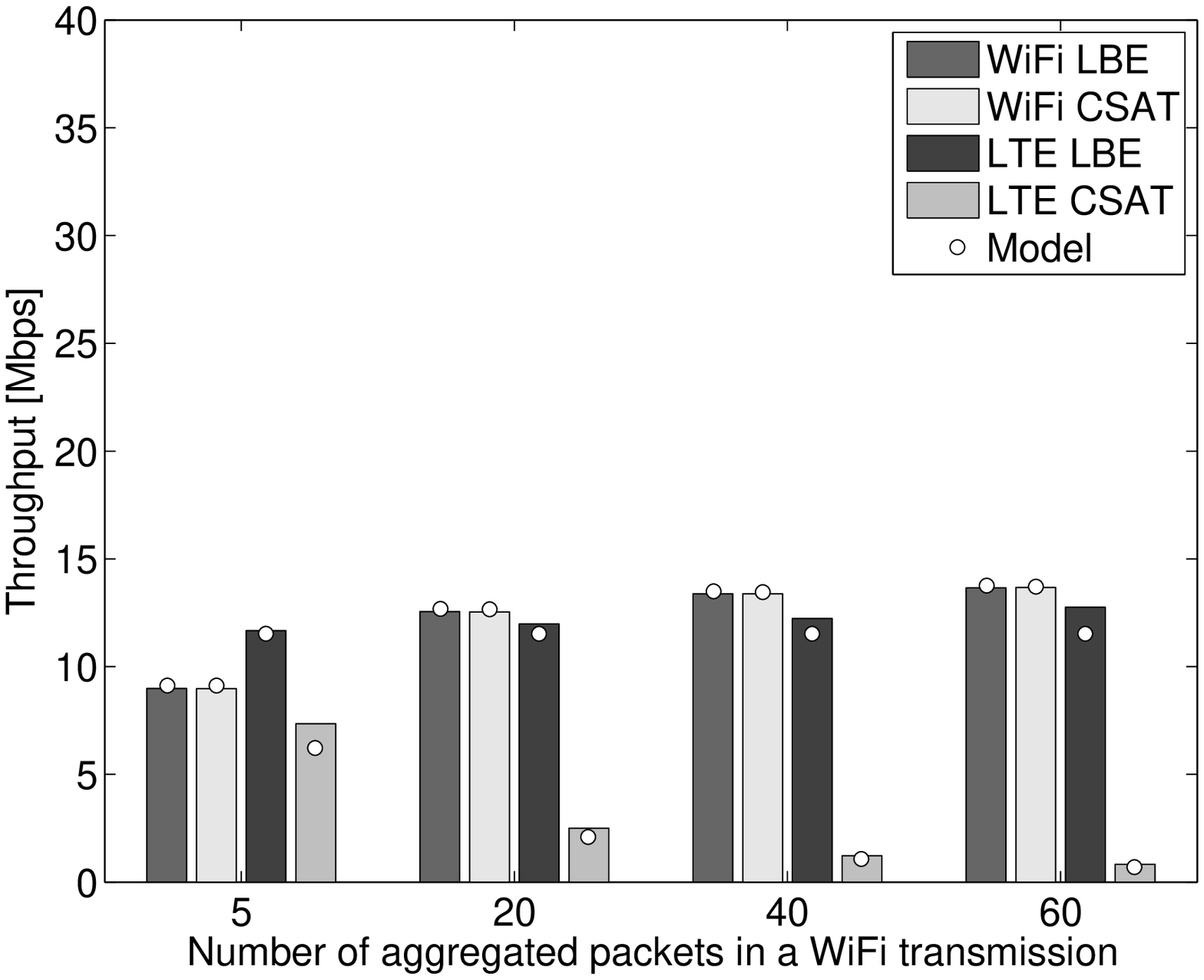}}\\
\subfigure[$n=1,T_{\rm on}=50$ms]{\includegraphics[width=0.65\columnwidth]{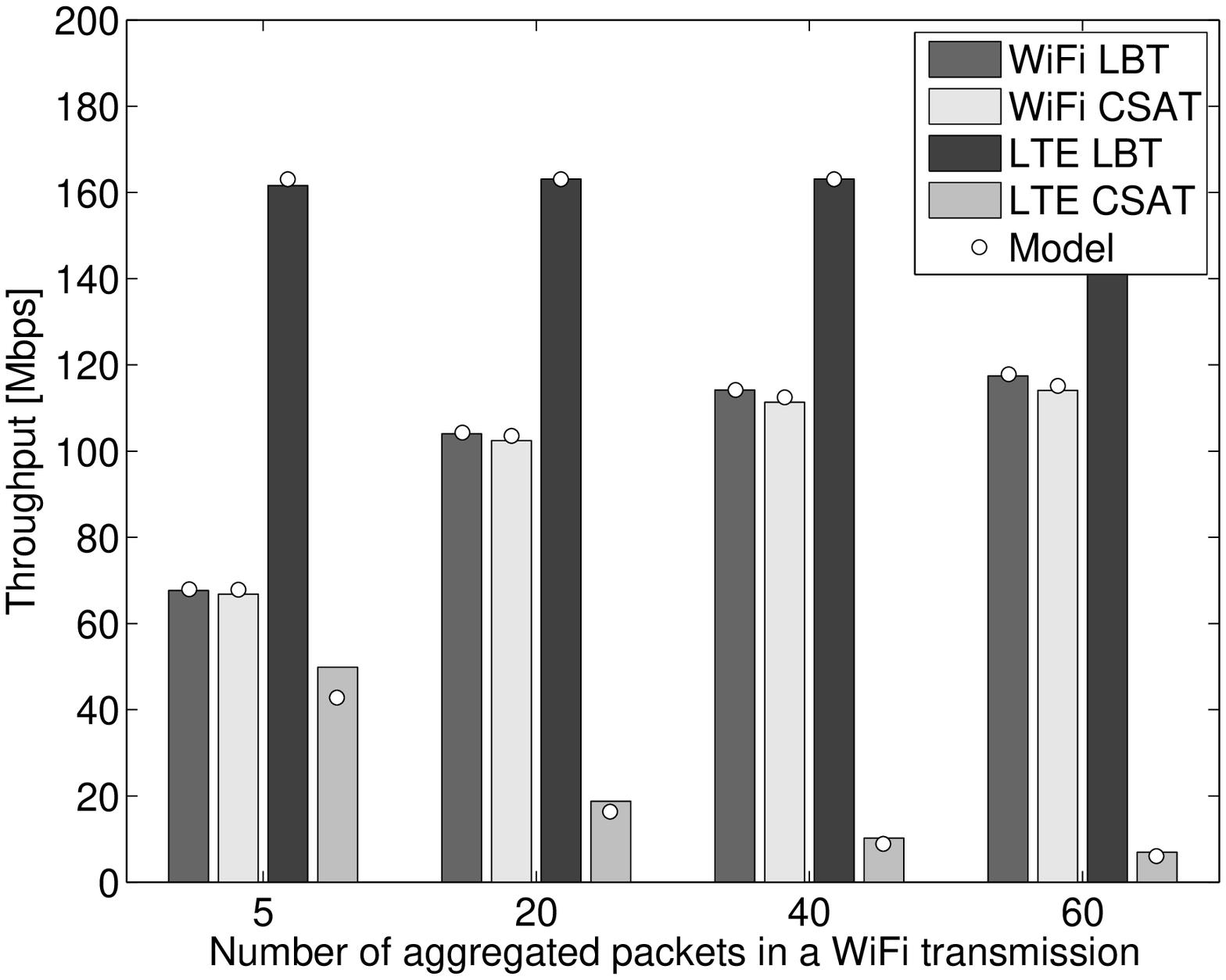}}
\subfigure[$n=3,T_{\rm on}=50$ms]{\includegraphics[width=0.65\columnwidth]{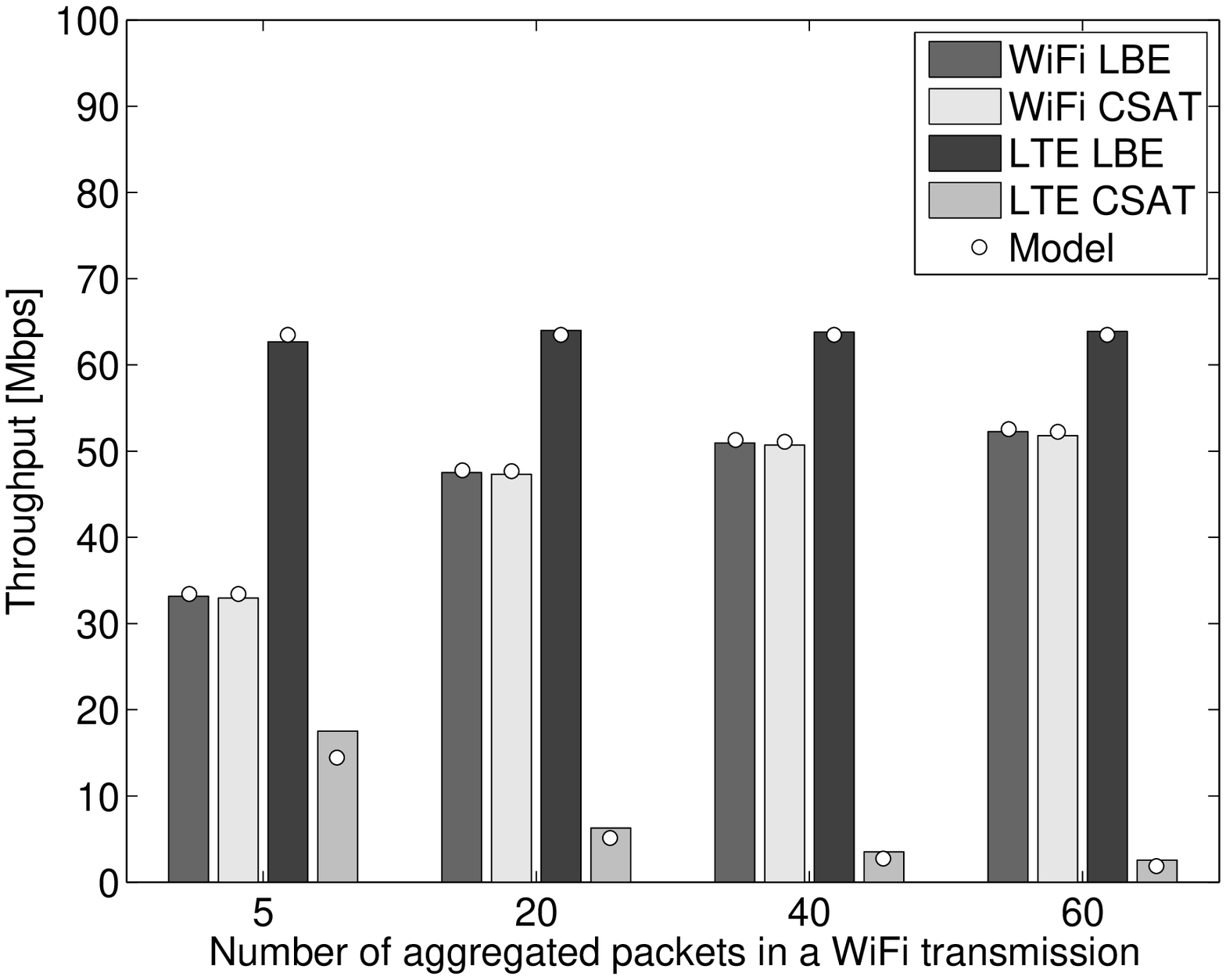}}
\subfigure[$n=9,T_{\rm on}=50$ms]{\includegraphics[width=0.65\columnwidth]{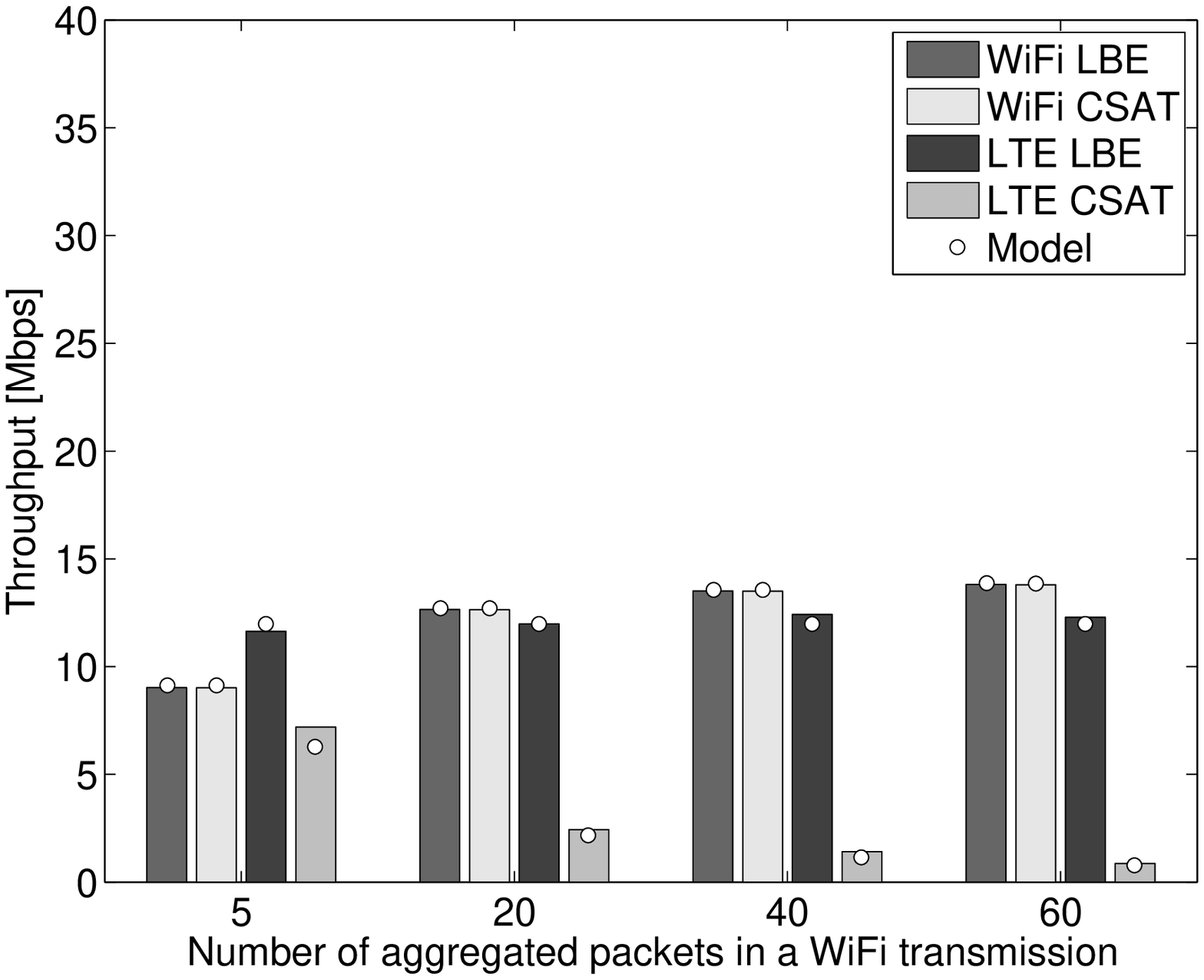}}
\caption{Throughput with imperfect carrier sensing for different configurations of $n$ and $T_{\rm on}$ while varying $n_{\rm agg}$ (effectively changing the packet size of WiFi transmissions). Simulation results are averages of $100$ simulation runs with $50$ s time horizon.}
\label{fig:throughput_prop_fair_nonidealcs}
\end{figure*}

\subsection{Explicit Communication}

The foregoing experimental measurements demonstrate the potential for much reduced sensitivity of random access carrier sensing when detecting scheduled transmissions rather than other random access transmissions, and therefore raise concerns regarding imperfect sensing of scheduled transmissions.  When physical carrier sensing of scheduled transmissions by random access transmitters is imperfect then interference (``collisions'') between scheduled and random access transmissions may increase substantially and so reduce network throughput and quality of service.   For example, in the extreme case where random access transmitters cannot detect scheduled transmissions at all then scheduled throughput is likely to be much reduced.   

One solution is to enable virtual carrier sensing via explicit communication between the scheduled and random access transmitters.  This might be achieved, for example, by modifying scheduled transmitters to transmit a signal decodable by the random access network at the start of a $T_{\rm on}$ period to announce the duration of the transmission e.g. the scheduled transmitter might send a WiFi CTS-to-self packet (a 802.11 frame defined for backward compatibility). The use of explicit communication using the CTS-to-self in the context of unlicensed LTE and WiFi has been proposed in~\cite{sadek2015extending}. Upon successfully receiving this signal, the CSMA/CA transmitters will defer transmissions until the start of the next $T_{\rm off}$ interval.  However, when such explicit communication is used for inter-network detection the explicit signaling is prone to loss.  For example this may occur when the scheduled and random access transmitters begin transmission simultaneously at the start of a $T_{\rm on}$ 
period, leading to a collision. In this case, the concern is that random access transmitters will then be incapable of correctly decoding the signal sent by the scheduled transmitter.      


\subsection{Throughput Model with Explicit Communication}\label{sec:appendix_model_imperfectcs}

In this section we extend the throughput model in Section~\ref{sec:appendix_model} to include explicit communication between the scheduled and random access transmitters at the start of each $T_{\rm on}$ period, e.g. via the scheduled transmitter sending a CTS-to-self.   The main difference from before is that we now need to take account of the fact that when a collision occurs at the start of a $T_{\rm on}$ period then the signal from the scheduled transmitter is lost and so the random access transmitters may continue transmitting during the $T_{\rm on}$ period rather than deferring to the scheduled transmission. For illustrative purposes we consider the physical carrier sensing by the CSMA/CA network to be completely ineffective to detect LTE transmissions.

\vspace{0.15cm}

\subsubsection{CSMA/CA Throughput}

We first note that the CSMA/CA throughput during a $T_{\rm on}$ period is the same regardless of signalling from the scheduled network.  Namely, either the CSMA/CA transmitters remain silent during the $T_{\rm on}$ period, and so no data is received, or the CSMA/CA transmissions during the $T_{\rm on}$ period collide with the scheduled transmissions and are lost, again with no data being received.   However, when carrier sensing is imperfect then there may now be a partial collision at the end of the $T_{\rm on}$ period that extends beyond the scheduled transmission duration, affecting the value of $\E[\hat{T}_{{\rm off}}]$ (the mean time during which successful CSMA/CA transmissions are possible).\newline

\underline{Preemptive Approach}

Assuming that \emph{(i)} on average the collisions at the end of the $T_{\rm on}$ period occur half-way through a CSMA/CA transmission, \emph{(ii)} collisions at the start of the $T_{\rm on}$ period are independent of collisions at the end of $T_{\rm on}$ (which should hold when $T_{\rm on} \gg \Delta$) and \emph{(iii)} these collisions occur with probability $p_{{\rm txA}} = \frac{(1-p_{\rm e})\Delta }{\E[M]}$, then:
\begin{align}
\E[\hat{T}_{{\rm off}}]=\bar{T}_{{\rm off}} - \frac{\Delta}{2}p_{{\rm txA}}(1-p_{{\rm txA}}) - \Delta p_{{\rm txA}}^2.\label{eq:newToff}
\end{align}

\underline{Opportunistic Approach}

Similarly as for the \emph{Preemptive} approach but considering that now partial collisions can only occur at the end of the $T_{\rm on}$ period, we have

\begin{align}
\E[\hat{T}_{{\rm off}}]=\bar{T}_{{\rm off}} - \frac{\Delta}{2}(1-p_{\rm e})p_{{\rm txA}},
\end{align}
where once again $p_{{\rm txA}} = \frac{(1-p_{\rm e})\Delta }{\E[M]}$.\newline

\vspace{0.15cm}

\subsubsection{Scheduled Network Throughput}
For simplicity we assume that the maximum idle space left between random access transmissions is smaller than the duration of a scheduled slot $\delta$. Thus, when a collision occurs at the beginning of a $T_{\rm on}$ period then all of the scheduled frames in a slot are lost and we have the following.\newline

\underline{{Preemptive Approach}}

We transmit successfully during $T_{\rm on}$ only in case of no collision with a CSMA/CA station:

\begin{align}
s_{{\rm txA}} =r\frac{T_{\rm on}(1-p_{{\rm txA}})} {T_{\rm on}+\bar{T}_{{\rm off}}}.
\end{align}

\underline{{Opportunistic Approach}} 

Only the duration $T_{\rm on} - T_{\rm res}$ adds to throughput as:

\begin{align}
s_{{\rm txA}} &=r\frac{(T_{\rm on} - T_{\rm res})(1-p_{{\rm txA}})}{T_{\rm on}+\bar{T}_{{\rm off}}},
\end{align}
with $p_{{\rm txA}} = 1-p_{\rm e}$ as in Section~\ref{sec:appendix_model}. Note that when there is not a collision with a CSMA/CA node then $c_2=T_{\rm res}$.

\subsection{Example: LTE and WiFi with Virtual Carrier Sensing}

%
%
%
Using CTS-to-self as signalling approach, Fig.~\ref{fig:throughput_prop_fair_nonidealcs} shows the proportional fair result (as in Theorem~\ref{tm:one}) obtained using simulations and the analysis of throughput presented above.  This can be compared with Fig.~\ref{fig:throughput_prop_fair} but with the difference that WiFi now only defers to LTE when no LTE/WiFi collision occurs at the start of the $T_{\rm on}$ period, i.e. assuming that WiFi only defers to LTE upon correct reception of the CTS-to-self.  It can be seen in Fig.~\ref{fig:throughput_prop_fair_nonidealcs}, that the throughput of WiFi remains practically unchanged.  However, the LTE throughput is severely penalised for all configurations when using CSAT.  The reason for this is the considerably higher collision probability of CSAT than LBE for the cases evaluated.  Note also that LTE throughput is reduced compared to Fig.~\ref{fig:throughput_prop_fair} for LBE when $n=9$ due to the increase of the collision probability with the number of WiFi 
stations.  Although the throughput in this case is not as reduced as with CSAT, the performance degradation is considerable. One way to alleviate this in LBE might be for the LTE to transmit before the DIFS to allow WiFi stations to decode the CTS-to-self, similarly to the approach described in~\cite{valls-lteu-wifi}, but this is outside the scope of the present paper.





%% file: scope.tex
\section{Scope}\label{sec:scope}

In our analysis we have made a number of assumptions, many of which can be fairly readily relaxed.  
%
%
\emph{(i) Loss-free channel}: Extension of our model to include channel losses is straightforward.  Namely, by {reducing} the success probability with a packet loss probability.  \emph{(ii) Multiple Channels/Channel Bonding}:  Both the scheduled and the random access networks may in general transmit across multiple channels.   However, provided they occupy disjoint channels, we can solve the allocation problem separately for each set of channels using the model in Section \ref{sec:appendix_model}.  That is, although we focus on a single channel here, the generalisation to multiple channels is immediate. \emph{(iii) Unsaturated Scheduled Network:} Extension of the analysis in Section \ref{sec:unsat} to the case in which the scheduled network is not saturated is straightforward if we ignore buffer dynamics and consider both $T_{\rm on}$ and $T_{\rm off}$ as random variables bounded by the offered load. 
    
We have also made a number of assumptions which are less easy to relax.  \emph{(i) Completely Overlapping Channels}:  We have considered that the channel widths used by the coexisting networks completely overlap. The extension to smaller channel widths is not straightforward as it is not clear the level of interference that each technology will cause to one another when using heterogeneous and partially overlapping channel widths. Refer to~\cite{jian2015coexistence} for the case of unlicensed LTE and WiFi.  \emph{(ii) Capture}: Our model assumes that concurrent transmissions result in a collision and the inability of the receiver to decode the message.   The main difficulty with including the capture effect in our analysis (where some receivers may successfully decode a colliding transmission) lies in specifying a suitable physical layer model and so we leave this for future work. \emph{(iii) Hidden Terminals}:  Perhaps the most significant omission from our analysis is hidden terminals.     The basic difficulties here arise from the fact that  hidden terminals can start transmitting even when a transmission by another station has already been in progress for some time and that the times hidden terminals attempt transmission are coupled to the dynamics of the transmissions they overhear.  We therefore leave consideration of channel allocation with hidden terminals out of the scope of this work.   It is perhaps also worth noting here that the prevalence of severe hidden terminals in real network deployments presently remains unclear.  While it is relatively easy to construct hidden terminal configurations in the lab that exhibit gross unfairness, it may be that such configurations are less common in practice. \emph{(iv) Multiple Scheduled Networks:} We have considered a single scheduled transmitter which may represent the case of multiple coordinated scheduled users/networks. However, the case of uncoordinated scheduled transmitters (as might be the case when they belong to different operators) is challenging as results depend greatly on the extent of the desychronisation of their slot boundaries. Some works to coordinate scheduled transmitters in the context of WiMaX include~\cite{siddique2010spectrum,muhleisen2009ieee}.

%% file: conclusions.tex
\section{Final Remarks}\label{sec:conclusions}

In this paper we address the coexistence of scheduled and random access transmitters in the same frequency channel. We show that there is an inherent cost due to the heterogeneity in channel access approaches. This cost is a per-transmission one and can thus be alleviated by increasing the duration of scheduled transmissions at the expense of increasing the variability of the delay for random access transmissions. We derive the joint proportional fair rate allocation and demonstrate that in this the heterogeneity cost is accounted for the channel airtime of the scheduled transmissions while the inefficiency of random access is accounted for in the channel airtime of the random access network.  We extend this analysis to consider unsaturated random access stations as well as imperfect inter-technology detection.  We illustrate the application of our analysis to the fair coexistence of LTE and WiFi.   Importantly, we show that, when optimally configured, both CSAT and LBT/LBE, result in the same throughput to WiFi, providing significant new insight on the current controversy on their ability to provide fairness to WiFi. We also show that, in certain circumstances, the heterogeneity cost is higher for CSAT, and thus the resulting LTE throughput is lower when compared to use of LBT/LBE. We also show that in the case of imperfect inter-technology detection, the use of explicit communication is more problematic in CSAT due to the generally higher probability of loss compared to LBT/LBE.